\documentclass{article} 
\usepackage{graphicx} 

\usepackage{tikz}
\usepackage{amsmath}
\usepackage{mathtools}
\usepackage{natbib}
\usepackage{caption}
\usetikzlibrary{arrows.meta,
                chains,
                decorations.pathreplacing,calligraphy,
                positioning,
                quotes,
                shapes.geometric
                }

\usetikzlibrary{matrix,decorations.pathreplacing}

\usepackage{xcolor}
\usepackage{easybmat}
\usepackage{multirow,bigdelim}
\usepackage{amsmath}
\usepackage{amsthm}
\usepackage{paralist}

\usepackage{amssymb}
\usepackage{algorithm}
\usepackage[justification = centering]{subcaption}
\usepackage{color}
\usepackage{graphicx}
\usepackage{wrapfig,epsfig}
\usepackage{epstopdf}
\usepackage{url}
\usepackage{graphicx}
\usepackage{color}
\usepackage{epstopdf}
\usepackage[noend]{algpseudocode}
\usepackage[T1]{fontenc}
\usepackage{bbm}
\usepackage{comment}
\usepackage{xspace}
\usepackage{thmtools}

\usepackage{tikz}
\usetikzlibrary{calc}
\usepackage[backref=page]{hyperref}  
\usepackage{cleveref}
\hypersetup{colorlinks=true,citecolor=blue,linkcolor=blue} 
\usetikzlibrary{arrows}
\usepackage[lmargin=1in,rmargin=1in,tmargin=0.8in,bmargin=0.8in]{geometry}

\newtheorem{theorem}{Theorem}[section]
\newtheorem{lemma}[theorem]{Lemma}
\newtheorem{definition}[theorem]{Definition}

\newtheorem{proposition}[theorem]{Proposition}
\newtheorem{corollary}[theorem]{Corollary}
\newtheorem{conjecture}[theorem]{Conjecture}

\newtheorem{fact}[theorem]{Fact}
\newtheorem{remark}[theorem]{Remark}

\newtheorem{example}[theorem]{Example}

\newcommand{\ffrac}[2]{#1/#2}
\newcommand{\envy}{\mathsf{envy}}

\newcommand{\wt}{\widetilde}

\newcommand{\eps}{\epsilon}

\newcommand{\R}{\mathbb{R}}

\newcommand{\Grid}{\mathrm{Grid}}

\renewcommand{\varepsilon}{\epsilon}
\renewcommand{\tilde}{\wt}

\DeclareMathOperator{\poly}{poly}

\newcommand{\UTP}{{\sc Unary 3-Partition}\xspace}
\newcommand{\TP}{{\sc 3-Partition}\xspace}

\renewcommand{\cal}[1]{\mathcal{#1}}
\newcommand{\GHA}{\textsc{Graphical House Allocation}}
\newcommand{\MLA}{\textsc{Minimum Linear Arrangement}}

\newcommand{\cw}{\mathsf{cutwidth}}
\newcommand{\pw}{\mathsf{pathwidth}}
\newcommand{\tw}{\mathsf{treewidth}}

\newcommand{\polylog}{\text{polylog}}

\definecolor{mygreen}{RGB}{80,180,0}
\definecolor{b2}{RGB}{51,153,255}
\definecolor{mycy2}{RGB}{255,51,255}

\newcommand{\Envy}{\mathsf{Envy}}

\colorlet{thechosenone}{red}
\colorlet{thechosentwo}{blue}
\makeatletter
\newcommand*{\RN}[1]{\expandafter\@slowromancap\romannumeral #1@}
\makeatother
\usepackage{lineno}

\newif\ifcomments
\commentstrue
\ifcomments
\newcommand{\rik}[1]{{\textcolor{red}{Rik: { #1}}}}
\else
\newcommand{\rik}[1]{}
\fi
\newif\ifcomments
\commentstrue
\ifcomments
\newcommand{\justin}[1]{{\textcolor{blue}{Justin: { #1}}}}
\else
\newcommand{\justin}[1]{}
\fi
\newif\ifcomments
\commentstrue
\ifcomments
\newcommand{\vignesh}[1]{{\textcolor{orange}{Vignesh: { #1}}}}
\else
\newcommand{\vignesh}[1]{}
\fi
\newif\ifcomments
\commentstrue
\ifcomments
\newcommand{\rohit}[1]{{\textcolor{magenta}{Rohit: { #1}}}}
\else
\newcommand{\rohit}[1]{}
\fi
\newif\ifcomments
\commentstrue
\ifcomments
\newcommand{\hadi}[1]{{\textcolor{purple}{Hadi: { #1}}}}
\else
\newcommand{\hadi}[1]{}
\fi
\newif\ifcomments
\commentstrue
\ifcomments
\newcommand{\andrew}[1]{{\textcolor{teal}{Andrew: { #1}}}}
\else
\newcommand{\andrew}[1]{}
\fi

\title{Tight Approximations for  Graphical House Allocation}
\author{%
  Hadi Hosseini\\
  Penn State University\\
  \texttt{hadi@psu.edu} \\
  \and
  Andrew McGregor \\
  UMass Amherst\\
  \texttt{mcgregor@cs.umass.edu} \\
  \and
  Rik Sengupta \\
  UMass Amherst\\
  \texttt{rsengupta@cs.umass.edu} \\
  \and
  Rohit Vaish \\
  IIT Delhi\\
  \texttt{rvaish@iitd.ac.in} \\
  \and
  Vignesh Viswanathan \\
  UMass Amherst\\
  \texttt{vviswanathan@umass.edu} \\
}
\date{}

\begin{document}

\maketitle

\begin{abstract}
The {\GHA} problem asks: how can $n$ houses (each with a fixed non-negative value) be assigned to the vertices of an undirected graph $G$, so as to minimize the ``aggregate local envy'', i.e., the sum of absolute differences along the edges of $G$? This problem generalizes the classical {\MLA} problem, as well as the well-known \emph{House Allocation Problem} from Economics, the latter of which has notable practical applications in organ exchanges. Recent work has studied the computational aspects of {\GHA} and observed that the problem is NP-hard and inapproximable even on particularly simple classes of graphs, such as vertex disjoint unions of paths. However, the dependence of any approximations on the structural properties of the underlying graph had not been studied.

In this work, we give a complete characterization of the approximability of {\GHA}. We present algorithms to approximate the optimal envy on general graphs, trees, planar graphs, bounded-degree graphs, bounded-degree planar graphs, and bounded-degree trees.
For each of these graph classes, we then prove \emph{matching} lower bounds, showing that in each case, no significant improvement can be attained unless P = NP. We also present general approximation ratios as a function of structural parameters of the underlying graph, such as treewidth; these match the aforementioned tight upper bounds in general, and are significantly better approximations for many natural subclasses of graphs. Finally, we present constant factor approximation schemes for the special classes of complete binary trees and random graphs. 

Some of the technical highlights of our work are the use of expansion properties of Ramanujan graphs in the context of a classical resource allocation problem, and approximating optimal cuts in binary trees by analyzing the behavior of consecutive runs in bitstrings.
\end{abstract}

\section{Introduction}


In the EconCS community, the \emph{House Allocation Problem} has been a topic of significant interest for some time \citep{shapley1974cores, svensson1999strategy, beynier2018localenvy, gsvfairhouse, kmsfairhouse}. In its canonical form, the problem involves a set of $n$ agents, a set of $n$ items (``houses''), and possibly different valuation functions for each agent. In general, given this framework, the problem asks for an ``optimally fair'' allocation of the houses to the agents. For instance, we might wish to minimize the total envy, or maximize the number of envy-free pairs of agents. In this context, as it is common in the fairness literature, an agent $i$ envies an agent $j$ in a particular allocation if according to agent $i$'s valuation function, the item received by agent $j$ is worth more than the item received by agent $i$; the amount of envy is the difference in these two values. The canonical problem has been studied in a variety of contexts, and is well-known as an algorithmically difficult problem to solve, for most reasonable fairness objectives.

\citet{canon} introduced a variant of the house allocation problem called {\GHA}. In this setting, there are $n$ agents, but now they are placed on the vertices of an undirected $n$-vertex graph $G = (V, E)$. There are still $n$ items with arbitrary values, but the agents are \emph{identical} in how they value these $n$ items (i.e., they all agree on the value of each house). {\GHA} now asks: how do we allocate each house to an agent so as to minimize the total envy along the edges of $G$?

We remark here that the setting where the agents are on a graph and only the envy along the graph edges is considered was studied before as well by \citet{beynier2018localenvy}, who considered ordinal preferences in such a setting, and were interested in maximizing the number of envy-free edges in the underlying graph. 

Observe that {\GHA} is a purely combinatorial problem: we are given an $n$-vertex graph $G = (V, E)$ and a multiset $H = \{h_1, \ldots, h_n\} \subseteq \R_{\geq 0}$. We wish to find the bijective function $\pi : V \to H$ that minimizes
$\sum_{(x, y) \in E}|\pi(x) - \pi(y)|$.

If the set of values were $H = \{1, \ldots, n\}$, then {\GHA} would be identical to the well-known {\MLA} problem. This was observed by \citet{canon}, who went on to show some remarkable differences between the two problems. For instance, while all hardness results carry over from {\MLA} to {\GHA}, the latter is actually a significantly harder problem even on very simple graphs. In particular, in {\MLA}, we can assume without loss of generality that the underlying graph is connected; this is because an optimal solution is given by taking each connected component separately, and optimally assigning a contiguous subset of values to it. We lose this guarantee in {\GHA}, even for  small graphs with just two connected components. As a typical example of the differences between the two problems, observe that if the underlying graph is a disjoint union of paths, then solving {\MLA} optimally takes linear time, but even on these simple instances, {\GHA} is NP-complete \citep{canon}.

We do note, however, that all the hardness constructions by \citet{canon} used the disconnectedness of the underlying graphs crucially, in finding reductions from bin packing instances. Their results also show that for very simple classes of disconnected graphs, {\GHA} is inapproximable to any finite factor. However, these proof techniques do not carry over to \emph{connected} graphs, and so it was not known whether any of these reductions would go through for connected graphs. For instance, a well known result by \citet{mlatrees} states that {\MLA} is solvable in polynomial time on trees; the complexity of this problem for {\GHA} was open.



\subsection{Our Contributions}

We present a complete characterization of the approximability of {\GHA} on various classes of connected graphs, summarized in Table \ref{table:summary1}. In particular, for any instance on the following graph classes, we show a polynomial-time\footnote{ In all our results, $\tilde{O}$ hides $\polylog(n)$ factors} algorithm on an $n$-vertex graph $G$ (with maximum degree $\Delta$) in that class for obtaining the stated multiplicative approximation to the optimal envy, and then demonstrate a matching lower bound that shows that any polynomial improvement on the approximation ratio is impossible on that graph class unless P = NP:
\begin{itemize}
    \item If $G$ is any connected graph, \emph{any} allocation attains the trivial upper bound of $O(n^2)$ (Proposition \ref{prop:trivialgeneral}). In Theorem \ref{thm:general-approx-lower-bound}, we show that we cannot have an $O(n^{2 - \epsilon})$-approximation for any $\epsilon > 0$. We also give a polynomial-time $\tilde{O}(\tw(G)\cdot\Delta)$-approximation algorithm (Corollary \ref{cor:cutwidth-upperbounds}).
    
    \item If $G$ is a tree, \emph{any} allocation attains the trivial upper bound of $O(n)$ (Proposition \ref{prop:trivialgeneral}). In Theorem \ref{thm:trees-approx-lower-bound}, we show that we cannot have an $O(n^{1 - \epsilon})$-approximation for any $\epsilon > 0$. This is in stark contrast to {\MLA}, where there are sub-quadratic algorithms for \emph{exact} solutions on trees \citep{mlatrees}. We also explicitly show a simple divide-and-conquer procedure (Algorithm \ref{alg:treelogn}) that gives the same $O(\Delta\log n)$-approximation in $O(n\log n)$ time.

    \item If $G$ is planar, Corollary \ref{cor:cutwidth-upperbounds} gives us a polynomial-time algorithm to achieve an $\tilde{O}(\sqrt{n\Delta})$-approximation to the optimal envy. In the worst case, $\Delta = \Theta(n)$, so this is a worst-case approximation of $\tilde{O}(n)$. Once again, Theorem \ref{thm:trees-approx-lower-bound} shows that we cannot have an $O(n^{1 - \epsilon})$-approximation for any $\epsilon > 0$.

    \item If $G$ is a bounded-degree graph, Corollary \ref{cor:cutwidth-upperbounds} gives us a polynomial-time algorithm to achieve an $\tilde{O}(\tw(G))$-approximation to the optimal envy. Again, this is a worst-case approximation of $\tilde{O}(n)$. Using Theorem \ref{thm:bounded-degree-approx-lower-bound}, we show that we cannot have an $O(n^{1 - \epsilon})$-approximation for any $\epsilon > 0$. This is our 
    most involved technical result, and it uses expansion properties of Ramanujan graphs.

    \item If $G$ is a bounded-degree planar graph, Corollary \ref{cor:cutwidth-upperbounds} gives us a polynomial-time $\tilde{O}(\sqrt{n})$-approximation algorithm. We match this by showing that we cannot have an $O(n^{0.5 - \epsilon})$-approximation for any $\epsilon > 0$ (Theorem \ref{thm:bounded-degree-planar-approx-lower-bound}).

    \item If $G$ is a bounded-degree tree, both Algorithm \ref{alg:treelogn} and Corollary \ref{cor:cutwidth-upperbounds} give us a polynomial-time algorithm that outputs an $\tilde{O}(1)$-approximation to the optimal envy. We show that finding the exact optimal envy is NP-hard (Theorem \ref{thm:bounded-degree-trees-np-complete}).
\end{itemize}

\begin{table*}
    \centering
    \def\arraystretch{1.3}
    \begin{tabular}{ |l|l|c|  }
 \hline
 \multicolumn{3}{|c|}{Approximations for {\GHA}} \\
 \hline
 \bf Graph Class & \bf Upper Bound & \bf Lower Bound \\
 \hline
\multirow{2}{15em}{Connected graphs}  & $O(n^2)$ (Prop.~\ref{prop:trivialgeneral}) & \multirow{2}{10em}{\centering $\omega(n^{2 - \varepsilon})$ (Thm.~\ref{thm:general-approx-lower-bound})} \\
& $O(\tw(G)\cdot\Delta\log^{2.5} n)$ (Cor.~\ref{cor:cutwidth-upperbounds}(\ref{cor:cwgeneral})) & \\
\hline 
\multirow{2}{15em}{Trees}  & $O(n)$ (Prop.~\ref{prop:trivialgeneral}) & \multirow{2}{10em}{\centering $\omega(n^{1 - \varepsilon})$ (Thm.~\ref{thm:trees-approx-lower-bound})} \\
& $O(\Delta\log n)$ (Alg.~\ref{alg:treelogn}, Cor.~\ref{cor:cutwidth-upperbounds}(\ref{cor:cwtrees})) & \\
\hline
Planar graphs  & $O(\sqrt{n\Delta}\log^{1.5}n)$ (Cor.~\ref{cor:cutwidth-upperbounds}(\ref{cor:cwplanar})) & $\omega(n^{1 - \varepsilon})$ (Thm.~\ref{thm:trees-approx-lower-bound}) \\
\hline
Bounded-degree graphs  & $O(\tw(G)\cdot\log^{2.5}n)$ (Cor.~\ref{cor:cutwidth-upperbounds}(\ref{cor:cwgeneral})) & $\omega(n^{1 - \varepsilon})$ (Thm.~\ref{thm:bounded-degree-approx-lower-bound}) \\
\hline
Bounded-degree planar graphs &  $O(\sqrt{n}\log^{1.5}n)$ (Cor.~\ref{cor:cutwidth-upperbounds}(\ref{cor:cwplanar})) & $\omega(n^{0.5 - \varepsilon})$ (Thm.~\ref{thm:bounded-degree-planar-approx-lower-bound}) \\
\hline
Bounded-degree trees & $O(\log n)$ (Thm.~\ref{thm:treelogn}, Cor.~\ref{cor:cutwidth-upperbounds}(\ref{cor:cwtrees})) & $> 1$ (NP-hard, Thm.~\ref{thm:bounded-degree-trees-np-complete}) \\
 \hline \hline
 Random graphs & $1+ O(\sqrt{\ln (n)/n})$ (Thm.~\ref{thm:random})) \emph{w.h.p.} & -- \\
 \hline
 Complete binary trees & $3.5$ (Thm.~\ref{thm:inorder})) & \textcolor{red}{open} (Conj.~\ref{conj:completebintrees}) \\
 \hline
\end{tabular}
\caption{Summary of our results. Here, $\Delta$ is the maximum degree of the graph in question, and the lower bounds assume P $\neq$ NP. Note that in all cases, the upper and lower bounds match up to polylogarithmic factors, showing that nontrivial improvements to these upper bounds are impossible unless P = NP. All our upper bounds are polynomial time.}
\label{table:summary1}
\end{table*}

Note that assuming connectivity in the results above is necessary, since \citet{canon} showed that disconnected graphs cannot have the optimal envy approximated to any finite factor. We give the first known results for connected graphs.

We also show that for random graphs, any allocation is a $(1 + o(1))$-approximation with high probability (Theorem \ref{thm:random}).

Finally, we investigate complete binary trees in further detail. We first show that the class of binary trees is not ``well-behaved'', by refuting a conjecture by \citet{canon} about the structural properties of exact optimal allocations on binary trees by means of a counterexample (Section \ref{sec:boundeddegreetrees}). The hardness results in Theorems \ref{thm:trees-approx-lower-bound}, \ref{thm:bounded-degree-approx-lower-bound}, and \ref{thm:bounded-degree-trees-np-complete} might have suggested that complete binary trees cannot have $o(\log n)$-approximations in general. We show, however, that just the in-order traversal on a complete binary tree achieves a $3.5$-approximation to the optimal envy (Theorem \ref{thm:inorder}). We also show that this approximation ratio cannot be improved beyond 1.67 by a natural class (``value-agnostic'') of algorithms. 



Our paper is organized as follows. In Section \ref{sec:prelims}, we set up preliminaries. In Sections \ref{sec:upper} and \ref{sec:lower}, we present our upper and lower bounds respectively from Table \ref{table:summary1}. In Section \ref{sec:completebintrees}, we discuss binary trees. We finish with concluding remarks and open directions in Section \ref{sec:conclusions}.

\subsection{Other Related Work} 
Our work is very close to the large body of results on the computability of {\MLA}. While finding optimal linear arrangements is intractable in general \citep{mlabinaryhard}, there have been several papers presenting approximation algorithms for the problem \citep{richarao2005mla,feige2007mla, even200mla}, with the best known approximation ratio being $O(\sqrt{\log n}\log{\log n})$ \citep{feige2007mla}. Note that it is relatively straightforward to show that an $\alpha$-approximation algorithm for the \MLA~problem yields an $\alpha\phi$ approximation for the $\GHA$ problem where $\phi=\max_{1\leq i\leq n-1} (h_{i+1}-h_i)/\min_{1\leq i\leq n-1} (h_{i+1}-h_i)$.

Our problem also generalizes the classical problem of \textsc{Minimum Bisection}, which asks how to partition a graph $G$ into two almost equally-sized components with the smallest number of edges going across the cut. This problem is NP-complete \citep{mlahard} and it is also known to be inapproximable by an additive factor of $n^{2-\epsilon}$ \citep{bj92}. These lower bounds carry over to the \GHA{} problem as well, although the latter is strictly harder. For instance, \textsc{Minimum Bisection} is known to be solvable exactly in polynomial time for forests, but {\GHA} is NP-hard~\citep{canon}.

The canonical house allocation problem has also been well-studied in the literature. Recall that, in the canonical house allocation problem, agents are allowed to disagree on the values of the houses. In this setting, the existence and computational complexity of envy-free allocations on graphs have been reasonably well-studied \citep{beynier2018localenvy,eiben2020parameterized,bredereck2022envy}, with the problem, unsurprisingly, being computationally intractable in most settings. 
There have also been a few papers studying the complexity of minimizing various notions of envy when the underlying graph is {\em complete} \citep{gsvfairhouse, kamiyama2021envy, aigner2022envy,MMS23complexity}.

\section{Model and Preliminaries}\label{sec:prelims}


We have a set of $n$ {\em agents} $V = [n]$ placed on the vertices of an undirected graph $G = (V, E)$. 
There are $n$ {\em houses}, each with a nonnegative \emph{value}, that need to be allocated to the agents. We represent the houses simply by the multiset of values $H = \{h_1, \ldots, h_n\}$, and assume WLOG that $h_1 \leq \ldots \leq h_n$. We will interchangeably talk about the house with value $h_i$ and the real number $h_i$. The pair $(G, H)$ defines an \emph{instance} of {\GHA}. 

An {\em allocation} $\pi: V \rightarrow H$ is a bijective mapping from agents (or nodes) to house values. Given an allocation $\pi$ and an edge $(i, j) \in E$, we define the {\em envy} along the edge $(i, j)$ as $|\pi(i) - \pi(j)|$. Our goal in {\GHA} is to compute an allocation $\pi^\ast$ that {\em minimizes} the total envy 
along all the edges of $G$:
\begin{align*}
    \Envy(\pi, G) := \sum_{(i, j) \in E} |\pi(i) - \pi(j)|.
\end{align*}

We adopt the following definition from  \citet{canon} that provides a geometric representation to visualize allocations.

\begin{definition}[Valuation Interval]\label{def:valn_interval}
For an instance $(G, H)$ of {\GHA}, define the \emph{valuation interval} as the closed interval $\left[h_1, h_n\right] \subset \R_{\geq 0}$. For any allocation $\pi$, the envy along the edge $(i, j) \in E$ is exactly the length of the interval $[\pi(i), \pi(j)]$ (assuming $\pi(i) \leq \pi(j)$). We sometimes call the intervals $[h_i, h_{i+1}]$ for $1 \leq i \leq n - 1$ the \emph{smallest subintervals} of the valuation interval.
\end{definition}

 An optimal allocation $\pi^\ast$ would minimize the sum of the lengths of the intervals corresponding to each of its edges.
An allocation $\pi$ is \emph{$\alpha$-approximate} if $\Envy(\pi, G) \le \alpha\cdot\Envy(\pi^\ast, G)$.


Fix any arbitrary class $\cal G$ of graphs (we allow $\cal G$ to be a singleton set). We say an algorithm $\mathsf{ALG}_\mathcal{G}$ is \emph{defined} on $\cal G$ if $\mathsf{ALG}_\mathcal{G}$ is well-specified and outputs a valid allocation on every instance $(G, H)$ of {\GHA} with $G \in \mathcal{G}$. Such an algorithm $\mathsf{ALG}_\mathcal{G}$ is an \emph{$\alpha$-approximation} if for all instances $(G, H)$ of {\GHA} with $G \in \mathcal{G}$, $\mathsf{ALG}_{\mathcal{G}}$ always outputs an $\alpha$-approximate allocation. A $1$-approximation is an exact algorithm.


\begin{definition}[Value-Agnostic Algorithms]
\label{defn:valueagnostic}
An algorithm $\mathsf{ALG}_\mathcal{G}$ defined on a graph class $\mathcal{G}$ is \emph{value-agnostic} if on every input $(G, H)$ with $G \in \mathcal{G}$, $\mathsf{ALG}_\mathcal{G}$ returns the same allocation on all instances where the \emph{ordering} of house values is the same (in other words, the algorithm only requires the ordinal ranking and not the numerical values).
If the graph class $\cal G$ admits a value-agnostic $\alpha$-approximation algorithm, we say $\cal G$ is \emph{$\alpha$-value-agnostic}. Otherwise, it is \emph{$\alpha$-value-sensitive}.
\end{definition}

How can we re-frame existing results on {\GHA} in the light of Definition \ref{defn:valueagnostic}? \citet{canon} show that, unless P = NP, there is no $1$-approximation algorithm $\mathsf{ALG}_\mathcal{G}$ when $\cal G$ is the set of vertex-disjoint unions of paths, cycles, or stars. In contrast, they show that value-agnostic \emph{exact} algorithms exist when $\cal G$ is the set of paths, cycles, or stars.

Of course, value-agnostic $\alpha$-approximations are extremely powerful algorithms, as they can exploit the graph structure \emph{independently} of the values in the {\GHA} instance. As we would expect, value-agnostic $1$-approximations do not always exist, even on very simple graph classes and even if we allow for unlimited time. For instance, consider the graph consisting of the disjoint union of $K_2$ and $K_3$. Figure \ref{fig:value_agnostic_ex} shows that this graph does not admit an $\alpha$-value-agnostic algorithm for any finite $\alpha$.

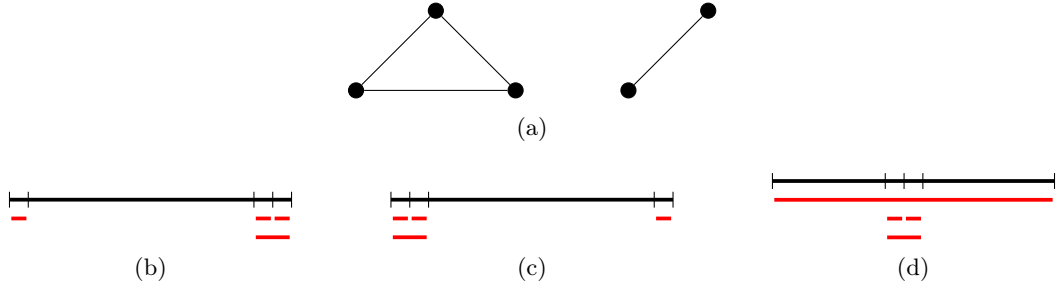
\begin{figure}[t]
    \centering
    \begin{subfigure}[b]{0.3\textwidth} 
    \centering
    \begin{tikzpicture}[,
    mycirc/.style={circle,fill=black, draw = black,minimum size=0.2cm,inner sep = 1.2pt}, node distance = 1.5cm]
    \node[mycirc] (center){};
    \node[mycirc, below left of=center] (left) {};
    \node[mycirc, below right of=center] (right){};
    \node[mycirc, right of=right] (p1) {};
    \node[mycirc, above right of=p1] (p2) {}; 
    \draw (p1)--(p2);
    \draw (center)--(left);
    \draw (center)--(right);
    \draw (left)--(right);
    \end{tikzpicture}
    \caption{}
    \label{subfig:two_edge_path}
    \end{subfigure}
    \\
    \vspace{10pt}
    \begin{subfigure}[b]{0.3\textwidth} 
    \centering
    \begin{tikzpicture}[xscale=0.5]
    \xdef\eps{0.05}
        \foreach \x in {1,2,14,15,16}
            {
                \coordinate (A\x) at ($(\x/2,0)$) {};
                \draw ($(A\x)+(0,3pt)$) -- ($(A\x)-(0,3pt)$);
            }
            \draw[line width=0.5 mm] (0.5,0) -- (8,0);
            \draw[color=thechosenone,line width=0.5 mm] (0.5+\eps,-0.25) -- (1-\eps,-0.25);
            \draw[color=thechosenone,line width=0.5 mm] (7+\eps,-0.25) -- (7.5-\eps,-0.25);
            \draw[color=thechosenone,line width=0.5 mm] (7.5+\eps,-0.25) -- (8-\eps,-0.25);
            \draw[color=thechosenone,line width=0.5 mm] (7+\eps,-0.5) -- (8-\eps,-0.5);
        \end{tikzpicture}
        \caption{}
    \label{subfig:p2_k3_1}
    \end{subfigure} 
    \begin{subfigure}[b]{0.3\textwidth} 
    \centering
    \begin{tikzpicture}[xscale=0.5]
    \xdef\eps{0.05}
        \foreach \x in {1,2,3,15,16}
            {
                \coordinate (A\x) at ($(\x/2,0)$) {};
                \draw ($(A\x)+(0,3pt)$) -- ($(A\x)-(0,3pt)$);
            }
            \draw[line width=0.5 mm] (0.5,0) -- (8,0);
            \draw[color=thechosenone,line width=0.5 mm] (0.5+\eps,-0.25) -- (1-\eps,-0.25);
            \draw[color=thechosenone,line width=0.5 mm] (1+\eps,-0.25) -- (1.5-\eps,-0.25);
            \draw[color=thechosenone,line width=0.5 mm] (7.5+\eps,-0.25) -- (8-\eps,-0.25);
            \draw[color=thechosenone,line width=0.5 mm] (0.5+\eps,-0.5) -- (1.5-\eps,-0.5);
        \end{tikzpicture}
        \caption{}
    \label{subfig:p2_k3_2}
    \end{subfigure} 
        \begin{subfigure}[b]{0.3\textwidth} 
        \centering
        \begin{tikzpicture}[xscale=0.5]
            \xdef\eps{0.05}

            \foreach \x in {1,7,8,9,16}
            {
                \coordinate (A\x) at ($(\x/2,-1.5)$) {};
                \draw ($(A\x)+(0,3pt)$) -- ($(A\x)-(0,3pt)$);
            }
            \draw[line width=0.5 mm] (0.5,-1.5) -- (8,-1.5);
            \draw[color=thechosenone,line width=0.5 mm] (0.5+\eps,-1.75) -- (8-\eps,-1.75);
            \draw[color=thechosenone,line width=0.5 mm] (3.5+\eps,-2) -- (4-\eps,-2);
            \draw[color=thechosenone,line width=0.5 mm] (4+\eps,-2) -- (4.5-\eps,-2);
            \draw[color=thechosenone,line width=0.5 mm] (3.5+\eps,-2.25) -- (4.5-\eps,-2.25);
    \end{tikzpicture}
    \caption{}
    \label{subfig:p2_k3_3}
    \end{subfigure}

    \caption{(a) shows the graph $K_2 \cup K_3$. (b), (c), and (d) show different valuation intervals with different optimal assignments for $K_2 \cup K_3$. In particular, note that every optimal solution for (b) requires the two smallest values to go to the $K_2$, but every optimal solution for (c) requires the three smallest values to go to the $K_3$, and an $\alpha$-value-agnostic algorithm cannot distinguish between the two, for any finite $\alpha$.}
    \label{fig:value_agnostic_ex}
\end{figure}

Although all our examples so far use the disconnectedness of the graphs to illustrate value-sensitivity, we will see in Section \ref{sec:completebintrees} that there are value-sensitive connected graphs as well.


For any graph $G = (V, E)$, and $S \subseteq V$, we denote by $\delta_G(S)$ the number of edges going across the cut $(S, V - S)$ in $G$. We will often estimate $\delta_G(S)$ for various subsets $S$. For $k \leq n-1$, we define $\delta_G(k) := \min_{|S| = k}\delta_G(S)$ as the size of the smallest cut in $G$ with $k$ vertices on one side. Of course, $\delta_G(k) = \delta_G(n - k)$ for all $k$. A $(k, n-k)$-cut in $G$ will be any cut $(S, V-S)$ with $|S| = k$.

We will use a few concepts from structural graph theory, most notably that of \emph{cutwidth}.

\begin{definition}\label{def:cutwidth}
    For a graph $G = (V, E)$ on $n$ vertices, let $\sigma = (v_1, \ldots, v_n)$ be any ordering of $V$. The \emph{width} of $\sigma$ is defined as
    \begin{equation*}
\mathsf{width}(\sigma, G):=        \max_{1 \leq \ell \leq n-1} \delta_G(\{v_1, \ldots, v_\ell\}).
    \end{equation*}
    The \emph{cutwidth} of $G$ is  the minimum width over all orderings of $G$, i.e.,
    \begin{equation*}
        \cw(G) := \min_{\sigma \in S_n}\mathsf{width}(\sigma, G)~.
    \end{equation*}
\end{definition}

The ordering $\sigma$ is often called a \emph{layout}. An \emph{optimal} layout is an ordering that achieves the cutwidth of $G$. The cutwidth is closely related to other standard notions of width used in structural graph theory. In particular, we have the following chain of inequalities (see \citet{korach1993cutwidth}):
\begin{align}
    &\tw(G)  \leq \pw(G) \leq \cw(G) \notag\\ &\qquad \leq O(\Delta\cdot\pw(G)) \leq O(\Delta\cdot\tw(G)\cdot\log n) 
    \label{eq:twpwcw}
\end{align}

Finding the exact cutwidth of $G$ in general is a difficult algorithmic problem. It can be computed exactly for trees (along with an optimal ordering) in time $O(n\log n)$ \citep{yannakakis1985treecutwidth}. However, even for planar graphs, the problem is NP-complete \citep{monien1988cutwidth}. 

If $G$ is sufficiently dense, there is a polynomial-time approximation scheme for the cutwidth \citep{frieze1996cutwidthdense}. In general, there is an efficient $O(\log^{1.5} n)$-approximation of the cutwidth known \citep{leighton1999cutwidthapprox}, which also returns a layout achieving this ratio. We will use this process as a subroutine several times in Section \ref{sec:upper}, for our upper bounds.
\section{Upper Bounds}\label{sec:upper}

The hardness of achieving optimal envy even on simple classes of graphs (e.g.,~disjoint unions of paths) \citep{canon} immediately gives rise to the question of whether we can approximate optimal solutions. As stated before, we need to assume connectivity in general.

We start by making a trivial observation (Proposition \ref{prop:trivialgeneral}): \textit{any} allocation of values to a connected graph is an $O(n^2)$-approximation to the optimal envy, and in fact an $O(n)$-approximation when the graph is a tree. This is due to the fact that every smallest subinterval of the valuation interval is covered by at most $|E|$ edges, but connectivity requires that it be covered by at least one edge. 

\begin{proposition}\label{prop:trivialgeneral}
    For \emph{any} instance of {\GHA} on a connected graph $G = (V, E)$, \emph{any} allocation is an $|E|$-approximation to the optimal value.
\end{proposition}

In what follows, we first discuss how to improve this bound for bounded-degree trees and then generalize this result to graphs based on a structural parameter called the \emph{cutwidth}. Finally, we showcase how our bounds can be significantly improved for the special class of random (Erd\H{o}s-Renyi) graphs.

\subsection{Trees}\label{sec:boundeddegreetrees}

In this section, we will present a recursive polynomial-time $O(\Delta\log n)$-approximation algorithm for any instance of {\GHA} where the underlying graph is any tree with maximum degree $\Delta$. Thus, for any tree with maximum degree $\Delta = o(n/\log n)$, our algorithm provides a better approximation than \Cref{prop:trivialgeneral}.

We will use the following folklore fact\footnote{ For a proof of this fact, see, for instance, \citet{mlatrees}, who attributes this as a folklore result to \citet{seidvasser}, who claims the fact is well-known, but proves it anyway.} without a proof.

\begin{fact}[Folklore]\label{fact:folklore}
    Every $n$-vertex tree $T$ has a \emph{center of gravity}: i.e., a vertex $v$ such that all connected components of $T - v$ have at most $n/2$ vertices. This vertex $v$ can be found in $O(n)$ time.
\end{fact}

\begin{algorithm}[t]
    \caption{Recursive Algorithm $\mathsf{TrickleDown}(T, H)$}
    \hspace{\algorithmicindent} 
    \textbf{Input:} {A {\GHA} instance on a tree $T$ and a set of values $H = \{h_1, \ldots, h_n\}$.} \\
    \textbf{Output:} {An $O(\Delta\log n)$-approximate allocation.}
    \begin{algorithmic}[1]
        \If{$|T| = 1$}
        \State \textbf{Allocate} the only house to the only vertex.
        \Comment{Base case}
        \Else
        \State Find a center of gravity $v$ of $T$.
        \State Let $T - v = T_1 + \ldots + T_k$, with $|T_i| = n_i$. 
        \Comment{$k \leq \Delta$, $n_i \leq n/2$.}
        \State Partition $H$ into the following contiguous sets: 
        \begin{align*}
            H_1 &= \{h_1, \ldots, h_{n_1}\}, \\
            H_2 &= \{h_{n_1+1}, \ldots, h_{n_1 + n_2}\} \\
            &\vdots \\
            H_k &= \{h_{n_1 + \ldots + n_{k-1} + 1}, \ldots, h_{n_1 + \ldots + n_k}\}. 
        \end{align*}
        \State \textbf{Allocate} $h_n$ to vertex $v$.
        \For{$i \in \{1, \ldots, k\}$}
            \State Recursively call $\mathsf{TrickleDown}(T_i, H_i)$.
        \EndFor
        \EndIf
    \State \Return the resulting allocation.
    \end{algorithmic}
    \label{alg:treelogn}
\end{algorithm}


We will use Fact \ref{fact:folklore} in developing a recursive algorithm (Algorithm~\ref{alg:treelogn}) that obtains an $O(\Delta\log n)$-approximation on trees. In each call, the algorithm first finds a center of gravity of the tree and subsequently uses this vertex to identify disjoint subtrees and solve the subproblems recursively on disjoint subintervals of the valuation interval.


\begin{restatable}{theorem}{thmtreelogn}\label{thm:treelogn}
There is an $O(n\log n)$-time algorithm that, given any instance on a tree with maximum degree $\Delta$, returns an allocation whose envy is at most $\Delta\log n$ times the optimal envy.
\end{restatable}
\begin{proof}

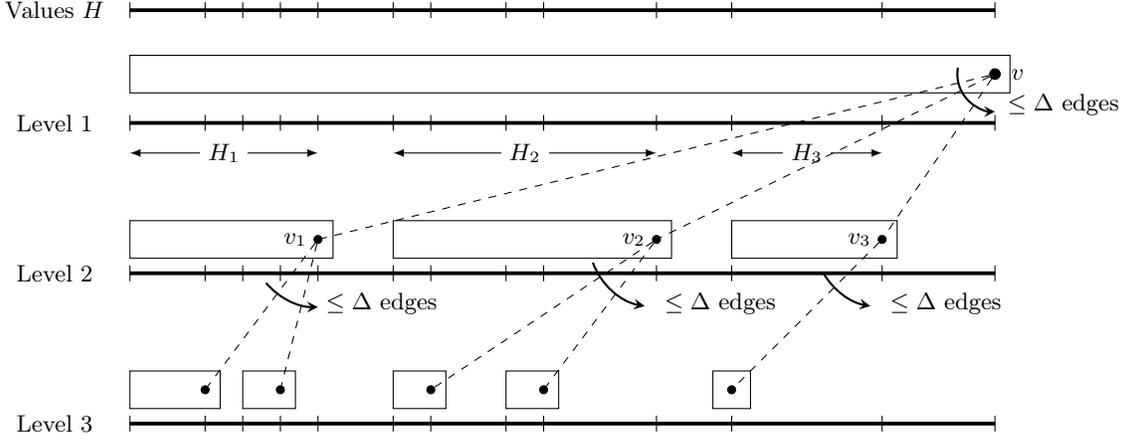
\begin{figure*}
    \centering
    \small
    \begin{tikzpicture}
    \xdef\eps{0.05}
    %
    %
        \foreach \x in {1,3,4,5,6,8,9,11,12,15,17,21,24}
            {
                \coordinate (A\x) at ($(\x/2,0)$) {};
                \draw ($(A\x)+(0,3pt)$) -- ($(A\x)-(0,3pt)$);
            }
            \node (10) at (-0.5,0) {Level 3};
            \draw[line width=0.5 mm] (0.5,0) -- (12,0);
            %
            %
            \draw[draw=black] (0.5,0.2) rectangle ++(1.2,0.5);
            \node[circle,draw=black,fill=black,inner sep=0pt,minimum size=3pt] (4) at (1.5,0.45) {};
            \draw[draw=black] (2,0.2) rectangle ++(0.7,0.5);
            \node[circle,draw=black,fill=black,inner sep=0pt,minimum size=3pt] (5) at (2.5,0.45) {};
            \draw[draw=black] (4,0.2) rectangle ++(0.7,0.5);
            \node[circle,draw=black,fill=black,inner sep=0pt,minimum size=3pt] (6) at (4.5,0.45) {};
            \draw[draw=black] (5.5,0.2) rectangle ++(0.7,0.5);
            \node[circle,draw=black,fill=black,inner sep=0pt,minimum size=3pt] (7) at (6,0.45) {};
            \draw[draw=black] (8.25,0.2) rectangle ++(0.5,0.5);
            \node[circle,draw=black,fill=black,inner sep=0pt,minimum size=3pt] (8) at (8.5,0.45) {};
            %
    %
    %
        \foreach \x in {1,3,4,5,6,8,9,11,12,15,17,21,24}
            {
                \coordinate (A\x) at ($(\x/2,2)$) {};
                \draw ($(A\x)+(0,3pt)$) -- ($(A\x)-(0,3pt)$);
            }
            \node (10) at (-0.5,2) {Level 2};
            \draw[line width=0.5 mm] (0.5,2) -- (12,2);
            %
            %
            \draw[draw=black] (0.5,2.2) rectangle ++(2.7,0.5);
            \node[circle,draw=black,fill=black,inner sep=0pt,minimum size=3pt] (1) at (3,2.45) {};
            \node (11) at (2.7,2.45) {$v_1$};
            \draw[draw=black] (4,2.2) rectangle ++(3.7,0.5);
            \node[circle,draw=black,fill=black,inner sep=0pt,minimum size=3pt] (2) at (7.5,2.45) {};
            \node (12) at (7.2,2.45) {$v_2$};
            \draw[draw=black] (8.5,2.2) rectangle ++(2.2,0.5);
            \node[circle,draw=black,fill=black,inner sep=0pt,minimum size=3pt] (3) at (10.5,2.45) {};
            \node (12) at (10.2,2.45) {$v_3$};
            %
    %
    %
        \foreach \x in {1,3,4,5,6,8,9,11,12,15,17,21,24}
            {
                \coordinate (A\x) at ($(\x/2,4)$) {};
                \draw ($(A\x)+(0,3pt)$) -- ($(A\x)-(0,3pt)$);
            }
            \node (10) at (-0.5,4) {Level 1};
            \draw[line width=0.5 mm] (0.5,4) -- (12,4);
            \draw[latex-latex] (0.5,3.6) -- ++ (2.5,0) node[midway,fill=white]{$H_1$};
            \draw[latex-latex] (4,3.6) -- ++ (3.5,0) node[midway,fill=white]{$H_2$};
            \draw[latex-latex] (8.5,3.6) -- ++ (2,0) node[midway,fill=white]{$H_3$};
            \node[circle,draw=black,fill=black,inner sep=0pt,minimum size=4pt] (0) at (12,4.65) {};
            \node (12) at (12.3,4.65) {$v$};
            \draw[draw=black] (0.5,4.4) rectangle ++(11.7,0.5);
    %
    %
        \foreach \x in {1,3,4,5,6,8,9,11,12,15,17,21,24}
            {
                \coordinate (A\x) at ($(\x/2,5.5)$) {};
                \draw ($(A\x)+(0,3pt)$) -- ($(A\x)-(0,3pt)$);
            }
            \node (10) at (-0.5,5.5) {Values $H$};
            \draw[line width=0.5 mm] (0.5,5.5) -- (12,5.5);
            %
    %
    \draw[dashed] (0) to (1);
    \draw[dashed] (0) to (2);
    \draw[dashed] (0) to (3);
    \draw[dashed] (1) to (4);
    \draw[dashed] (1) to (5);
    \draw[dashed] (2) to (6);
    \draw[dashed] (2) to (7);
    \draw[dashed] (3) to (8);
    %
    %
    \def\circledarrow#1#2#3{
    \draw[#1,-stealth] (#2) +(-190:#3) arc(-190:-90:#3);
    }
    \circledarrow{thick}{0}{0.5cm};
    \node [below right=0.1cm and 0cm of 0] {$\leq \Delta$ edges};
    \def\circledarrow#1#2#3{
    \draw[#1,-stealth] (#2) +(-140:#3) arc(-140:-90:#3);
    }
    \node [below right of = 1, node distance=1.2cm] {\small{$\leq \Delta$ edges}};
    \circledarrow{thick}{1}{0.9cm};
    \def\circledarrow#1#2#3{
    \draw[#1,-stealth] (#2) +(-160:#3) arc(-160:-100:#3);
    }
    \node [below right of = 2, node distance=1.2cm] {\small{$\leq \Delta$ edges}};
    \circledarrow{thick}{2}{0.9cm};
    \def\circledarrow#1#2#3{
    \draw[#1,-stealth] (#2) +(-150:#3) arc(-150:-100:#3);
    }
    \node [below right of = 3, node distance=1.2cm] {$\leq \Delta$ edges};
    \circledarrow{thick}{3}{0.9cm};
    \end{tikzpicture}
    \caption{Visualization of Algorithm \ref{alg:treelogn}.}
    \label{fig:treelogn}
\end{figure*}

We will show that Algorithm \ref{alg:treelogn} provides the desired guarantee. The algorithm starts by locating the center of gravity $v$ of the given tree $T$ (which is guaranteed to exist by Fact \ref{fact:folklore}). Then it assigns the largest (i.e., rightmost) value to node $v$, and recursively constructs the assignment for each of the disjoint subtrees in $T - v$.
      
      It is easiest to visualize the allocation resulting from Algorithm \ref{alg:treelogn} as in Figure \ref{fig:treelogn}. All recursive calls in line 9 occur in a single ``level'' of the figure, and all subsequent recursive calls from the subtrees $T_1, \ldots, T_k$ can also be packed into a single level, as the edges in $T_i$ and the edges in $T_j$ do not overlap, for any $i \neq j$. The crucial point is that the envy incurred strictly within disjoint subtrees $T_i$ and $T_j$ cannot involve the same smallest subintervals of the original instance.

      Let us analyze the total envy in the final allocation that is output by Algorithm \ref{alg:treelogn}. There are at most $\Delta$ edges adjacent to $v$, and each of them incur their envy in level $1$ of Figure \ref{fig:treelogn}. Each edge gets an envy of at most $(h_n - h_1)$, and therefore, the total envy on these edges is at most $\Delta\cdot(h_n - h_1)$. The total envy along the edges adjacent to $v_1, \ldots, v_k$ (except the ones accounted for in the levels above) are at most $\Delta\cdot(h_{n_1} - h_1), \ldots, \Delta\cdot(h_{n_1+\ldots+n_k} - h_{n_1+\ldots+n_{k-1}+1})$. Since the subintervals are all disjoint, this level accounts for an envy of at most $\Delta\cdot(h_n - h_1)$ as well. We can continue this argument through the lower levels.

      How many levels are there? Because each vertex picked at each recursive call is a center of gravity of the next subtree, the size of each subtree is at most half the size of the subtree at its parent level. The number of levels, therefore, is at most $\log n$. This gives us a total envy of $\Delta\cdot(h_n - h_1)\cdot\log n$.

      Note that the optimum envy has to be at least $h_n - h_1$ for any connected graph. This gives us an approximation ratio of $\Delta\log n$.

      The bound on the running time also arises from Fact \ref{fact:folklore}, which ensures that line 4 can be done in time $O(n)$. For each subtree $T_i$ in level $1$, we can find a center of gravity in time $O(n_i)$, so the total amount of work done to find the centers of gravity at this level is $O(n_1) + \ldots + O(n_k) = O(n)$. Since this is summed over $\log n$ recursive levels, the total running time is $O(n\log n)$.
\end{proof}

We remark that with a slightly more careful analysis,\footnote{  Technically this involves tweaking the algorithm such that the center of gravity is assigned slightly differently in line 7, and the partition of $H$ is consistent with this.} we can improve the approximation ratio to $(1/2)\cdot(1 + \Delta +\Delta\log n)$. In particular, for any instance on a binary tree, the optimal envy can be $(2\log n)$-approximated in $O(n\log n)$ time. 

\subsection{Cutwidth}\label{sec:cutwidth}

In this section, we generalize the result from Section \ref{sec:boundeddegreetrees} using the structural graph theoretic property of cutwidth (Definition \ref{def:cutwidth}). This will enable us to have a black-box process to obtain envy approximations parameterized by the cutwidth. All of these algorithms will be value-agnostic.

\begin{restatable}{theorem}{thmcwgeneric}\label{thm:cwgeneric}
Let $(G, H)$ be a {\GHA} instance defined on a connected graph $G$. Given a layout $\sigma$ that $\beta$-approximates $\cw(G)$, we can efficiently construct an allocation $\pi$ that is a $(\beta \cdot \cw(G))$-approximation to the optimal envy.
\end{restatable}
\begin{proof}
We construct the allocation $\pi$ as follows: for each agent $i \in V$, if $\sigma(i) = j$, we set $\pi(i) = h_j$; that is, we give $i$ the $j$-th least-valued house. 
Since $G$ is connected, the total envy of any allocation is at least $(h_n - h_1)$. In the allocation $\pi$, the number of edges of $G$ spanning any smallest subinterval of the valuation interval is at most $\mathsf{width}(\sigma, G)$ by definition, and therefore the envy from $\pi$ is at most $\sum_{i = 1}^{n-1}\mathsf{width}(\sigma, G)\cdot(h_{i+1} - h_i) = \mathsf{width}(\sigma, G)\cdot(h_n - h_1)$. Hence, if $\sigma$ is an $\beta$-approximation for the cutwidth, then $\pi$ is an $\beta \cdot \cw(G)$-approximation to the optimal envy of $G$, as claimed.
\end{proof}

The next corollary follows from Theorem \ref{thm:cwgeneric} and Equation \ref{eq:twpwcw} when combined with existing bounds on the cutwidth, treewidth, or pathwidth \citep{korach1993cutwidth, leighton1999cutwidthapprox, djidjev2006cutwidth} of certain graph families along with the best known approximation results of these quantities \citep{yannakakis1985treecutwidth,leighton1999cutwidthapprox}.

\begin{corollary}\label{cor:cutwidth-upperbounds}
There exist polynomial-time value-agnostic approximation algorithms for the following classes:
\begin{enumerate}[(i)]
    \item An $O(\Delta \log n)$-approximation algorithm on trees,\label{cor:cwtrees}
    \item An $O(\sqrt{n \Delta} \log^{1.5} n)$-approximation algorithm on planar graphs,\label{cor:cwplanar}
    \item An $ O(\tw(G) \cdot \Delta \log^{2.5} n)$-approximation algorithm on general connected graphs.\label{cor:cwgeneral}
\end{enumerate}
%
%
\end{corollary}

Note that for each class of graphs listed above $\Delta$ can be $O(n)$ in the worst case, and for general connected graphs, $\tw(G)$ can be $O(n)$ in the worst case as well. So, in the worst case, the first and third results are asymptotically worse than the trivial bound given by \Cref{prop:trivialgeneral}. However, for many natural subclasses of these graphs, such as bounded-degree graphs and bounded-degree trees, \Cref{cor:cutwidth-upperbounds} yields strictly better approximation guarantees.

\subsection{Random Graphs}\label{subsec:random-graphs}
We next consider random graphs, specifically Erd\H{o}s-Renyi graphs, where $G\sim {\mathcal G}_{n,1/2}$ denotes a random graph on $n$ nodes where every edge is present with probability $1/2$ and all edges are independent. We show that \GHA{} on such graphs can be approximated up to a factor $1+o(1)$ regardless of the valuation interval. The central observation is that for any subset of nodes $S$, $\delta_G(S)$ is tightly concentrated around $|S|(n-|S|)/2$. 
\begin{restatable}{lemma}{randomlemma}\label{lem:random}
    For $G\sim {\mathcal G}_{n,1/2}$, \[\Pr \left [ \forall S\subseteq V, (1-\epsilon) \leq \frac{\delta_G(S)}{|S|(n-|S|)/2} \leq (1+\epsilon) \right ] \geq 1- \exp(-\Omega(\epsilon^2 n)) \ , \]
for any $\epsilon \geq \sqrt{24 \ln (n)/n}$.

\end{restatable}
\begin{proof}
The expected size of the cut $\delta_G(S)$ is $E[\delta_G(S)]=|S|(n-|S|)/2$. By applying the Chernoff bound, we obtain:
\[\Pr[|\delta_G(S)-E[\delta_G(S)]|\geq \epsilon E[\delta_G(S)]] \leq  2 \exp( -\epsilon^2 E[\delta_G(S)]/3) \ . \]
Hence, by the union bound the probability there exists a set $S$ of size $k \le n/2$ such that $|\delta_G(S)-E[\delta_G(S)]|\geq \epsilon E[\delta_G(S)]$ is at most
\begin{align*}    
2\exp( - \epsilon^2 k(n-k)/6) \binom{n}{k}
& \leq 2 \exp( - \epsilon^2 kn/12 + k \ln n) \\
& \leq 2 \exp( - \epsilon^2 kn/24),
\end{align*}
  assuming $\epsilon \geq \sqrt{24 \ln (n)/n}$. 
The lemma then follows by taking the union bound over $k$ and noting that \[\sum_{k=1}^{n/2} 2 \exp( - \epsilon^2 kn/24)=\exp(-\Omega(\epsilon^2 n)) \ .\qedhere \]
\end{proof}

Lemma \ref{lem:random}, with $\epsilon=\sqrt{24 \ln(n)/n},$ implies that with high probability, the cost of the optimum solution is at least \[\sum_{i=1}^{n-1} (h_{i+1}-h_{i}) \delta_G(i)\geq \sum_{i=1}^{n-1} (h_{i+1}-h_{i})(1-\epsilon) i(n-i)/2, \]
whereas the cost of an arbitrary allocation is at most  \[\sum_{i=1}^{n-1} (h_{i+1}-h_{i})(1+\epsilon) i(n-i)/2  . \]
Therefore an arbitrary allocation is a $(1+\epsilon)/(1-\epsilon)= 1+ O(\sqrt{\ln (n)/n})$-approximation.

\begin{theorem}\label{thm:random}
For  $G \sim \cal G_{n, 1/2}$, any allocation is a $1+ O(\sqrt{\ln (n)/n})$ approximation with probability at least $1 - 1/\poly(n)$.
\end{theorem}

\section{Lower Bounds}\label{sec:lower}

Every algorithm presented in Section \ref{sec:upper} is value-agnostic. It might seem reasonable to assume, therefore, that there are more powerful approximation schemes that exploit the numerical values in $H$ in some way. Indeed, our results on random graphs suggest that, for most graphs, we can do significantly better. Remarkably, we show in this section that this is \emph{not} the case, and our value-agnostic algorithms are strong enough to give us nearly optimal approximation guarantees. Specifically, we show inapproximability results matching our upper bounds (up to $\polylog$ factors) for every class of graphs considered. Our lower bounds will use reductions from the {\UTP} problem.

\begin{definition}[\TP]
Given a multiset of $3m$ naturals $A = \{a_1, \dots, a_{3m}\} \subseteq \mathbb{N}_{> 0}$ and a natural $T \in \mathbb{N}_{> 0}$ such that $\sum_{j \in [3m]} a_j = mT$, {\TP} asks whether $A$ can be partitioned into $m$ triplets $(S_1, S_2, \dots, S_m)$ such that the sum of each triplet is equal to $T$.
\end{definition}

The \TP problem is NP-complete even when all the inputs are given in unary and each item in $A$ is strictly between $\ffrac{T}{4}$ and $\ffrac{T}{2}$ \citep{garey1979computers}. We refer to this variant as {\UTP}. Note that {\UTP} is just a reformulation of \textsc{Bin Packing}: there are $3m$ integers that sum to $mT$, and we wish to fit these integers into $m$ bins each of capacity $T$. The condition of three integers in each bin is redundant, as it is implied by the constraint that each integer is strictly between $T/4$ and $T/2$.

Some of our results and proofs in this section (specifically Theorems \ref{thm:trees-approx-lower-bound} and \ref{thm:bounded-degree-planar-approx-lower-bound}) are very similar to results about the inapproximability of the {\em balanced graph partition} problem \citep{feldmann2012gridpartition, feldmann2015treepartition}. The rest of our proofs use novel gadgets and techniques.

\subsection{Trees and Planar Graphs}\label{sec:treelowerbound}
Recall that we presented two approximation guarantees for trees, $O(n)$ (Proposition \ref{prop:trivialgeneral}) and $O(\Delta \log n)$ (Corollary \ref{cor:cutwidth-upperbounds}). Both of these results are $\tilde{O}(n)$ in the worst case. 

\begin{restatable}{theorem}{thmtreesapproxlowerbound}\label{thm:trees-approx-lower-bound}
For any constant $\varepsilon > 0$, there is no efficient $O(n^{1-\varepsilon})$ approximation algorithm for {\GHA} on depth-$2$ trees unless P = NP.
\end{restatable}
\begin{figure}[ht]
    \centering
    \begin{tikzpicture}[scale=0.6,
    mycirc/.style={circle,fill=black, draw = black,minimum size=0.10cm,inner sep = 1.5pt},
    mycirc2/.style={circle,fill=white,minimum size=0.75cm,inner sep = 1.5pt},
    mycirc3/.style={circle,fill=black,draw=black,minimum size=0.10cm,inner sep = 1.5pt},
    level 1/.style={sibling distance=30mm},
    level 2/.style={sibling distance=10mm},
    BC/.style = {decorate,  
                     decoration={calligraphic brace, amplitude = 4mm, mirror},
                     very thick, pen colour={black}
                    }
    ]
    \node[mycirc, label=above left:{r}] {} 
        child {node[mycirc, label=left:{$x_1$}] {}
            child[solid] {node[mycirc3](x1left){}} 
            child[solid] {node[mycirc2]{$\dots$}} 
            child[solid] {node[mycirc3](x1right){}}}
        child {node[mycirc, label=left:{$x_2$}] {}
            child[solid] {node[mycirc3](x2left){}} 
            child[solid] {node[mycirc2]{$\dots$}} 
            child[solid] {node[mycirc3](x2right){}}}
        child {node[mycirc2] {$\dots$}}
        child {node[mycirc, label=right:{$x_{3m}$}] {}
            child[solid] {node[mycirc3](xmleft){}} 
            child[solid] {node[mycirc2]{$\dots$}} 
            child[solid] {node[mycirc3](xmright){}}};
    
    \draw[BC]   ([yshift = -3mm] x1left.south west) coordinate (aux) -- 
                    node[midway,below=5mm]{$Ca_1 - 1$}
            (x1right.south east|- aux);
    \draw[BC]   ([yshift = -3mm] x2left.south west) coordinate (aux) -- 
                    node[midway,below=5mm]{$Ca_2 - 1$}
            (x2right.south east|- aux);
    \draw[BC]   ([yshift = -3mm] xmleft.south west) coordinate (aux) -- 
                    node[midway,below=5mm]{$Ca_{3m} - 1$}
            (xmright.south east|- aux);
    \end{tikzpicture}
    \caption{Mapping a \UTP instance to a tree.}
    \label{fig:trees-reduction-tree}
\end{figure}
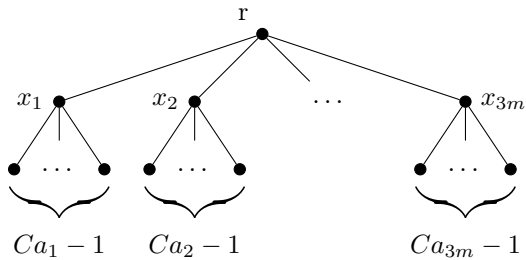



\begin{proof}
We give a reduction from \UTP. For some constant $\varepsilon > 0$, assume there is an efficient $O(n^{1- \varepsilon})$ approximation algorithm $\mathsf{ALG}_{\cal G}$ where $\cal G$ corresponds to the class of depth-$2$ trees. In other words, there is a constant $\gamma$ such that for all instances $(G, H)$ with $G \in \mathcal{G}$, $\mathsf{ALG}_{\cal G}$ outputs an allocation with total envy within a multiplicative factor of $\gamma n^{1- \varepsilon}$ to the optimal envy.

Given an instance of \UTP, we construct an instance $(G, H)$ of {\GHA} as follows: The graph $G$ is a rooted depth-$2$ tree where the root $r$ of the tree has $3m$ children $\{x_1, \dots, x_{3m}\}$. Each of these nodes $x_i$ has $C a_i - 1$ children (see Figure \ref{fig:trees-reduction-tree}). Here, $C$ is a positive integer whose exact value we shall determine later.

The total number of nodes in $G$ is $1 + \sum_{i \in [3m]} Ca_i = 1 + CmT$, and so we must specify $1 + CmT$ house values in $H$. We define $H$ with $CT$ values of $j$ for $j \in [m]$, together with a single value of $0$. Note that this construction can be done in polynomial time as long as $C$ is polynomially large, since the input to the $3$-partition instance is given in unary. 

We show that, for an appropriate choice of $C$, the minimum total envy output by $\mathsf{ALG}_{\cal G}$ for the instance $(G, H)$ is at least $\left \lceil (12\gamma m^3 T + 1)^{\ffrac{1}{\varepsilon}} \right \rceil$ if and only if there exists {\em no} valid partition for the original \textsc{3-Partition} instance, i.e., the original instance was a NO instance.

$(\Rightarrow)$ Assume there is a valid $3$-partition $(S_1, \ldots, S_m)$ for the original instance. We denote this $3$-partition using a mapping $\mu : A \to \{S_1, \ldots, S_m\}$ that maps each number in the multiset $A$ to one of the triplets $S_j$. We construct an allocation for the graph $G$ as follows: for any item $i$, if $\mu(i) = S_j$, we allocate houses with value $j$ to $x_i$ and all of its children. We finally allocate the house with value $0$ to the root. Note that this is a valid allocation: for any $j \in [m]$, we allocate exactly $\sum_{i \in [3m]: \mu(i) = S_j} Ca_i = CT$ houses of value $j$. 

We can easily upper bound the envy of this allocation; this upper bound also serves as an upper bound for the minimum total envy for the instance $(G, H)$. There is no envy between any $x_i$ and any of its children. So the only edges with potential envy are the ones incident on the root. There are $3m$ such edges, each incurring envy at most $m$. This gives us an upper bound of $3m^2$ on the total envy. This implies, when there is a valid $3$-partition, the approximation algorithm $\mathsf{ALG}_{\cal G}$ will output an allocation with envy at most $3m^2\gamma(1 + CmT)^{1-\varepsilon}$.

$(\Leftarrow)$ Assume there is no valid $3$-partition in the original instance. We will show that any allocation has a total envy of at least $C$. We do this by examining the houses allocated to the depth-1 nodes $\{x_1, \dots, x_{3m}\}$. If any $x_i$ is allocated a value of $0$ in some allocation $\pi$, then the allocation $\pi$ has a total envy of at least $C$ since any $x_i$ has at least $C -1$ children and $1$ parent receiving a value of at least $1$ each. 

So now assume no $x_i$ is allocated a value of $0$ in $\pi$. For notational convenience, assume WLOG that $x_1, \dots, x_{\ell_1}$ are allocated a value $1$, $x_{\ell_1 +1}, \dots, x_{\ell_1 + \ell_2}$ are allocated a value $2$ and so on. If for all $j \in [m]$, $\sum_{h \in [\ell_j]} a_{\ell_0 + \ell_1 + \dots + \ell_{j-1} + h} = T$ (with $\ell_0 = 0$), then we violate our assumption that there is no valid $3$-partition in the original instance. Therefore there exists one $j \in [m]$ such that $\sum_{h \in [\ell_j]}a_{\ell_0 + \ell_1 + \dots + \ell_{j-1} +h} > T$. Assume again for notational convenience that $j = 1$. Since all values are integers, we can restate the inequality above as $\sum_{h \in [\ell_1]}a_{h} \ge T + 1$. This implies $\sum_{h \in [\ell_1]} Ca_{h} \ge CT + C$. 

Coming back to the allocation $\pi$, we have that $\sum_{h \in [\ell_1]} Ca_{h} \ge CT + C$ implies that there are at least $C$ nodes out of all the children of $\{x_{1}, x_2, \dots, x_{\ell_1}\}$ which are not allocated a value of $1$. The envy that each of these $C$ nodes will have towards their parents is at least $1$. This implies that the total envy of allocation $\pi$ is at least $C$. 

We set $C = \left \lceil (12\gamma m^3 T + 1)^{\ffrac{1}{\varepsilon}} \right \rceil$ to complete the reduction. When there is no valid $3$-partition, the total minimum envy (and therefore, the envy output by $\mathsf{ALG}_{\cal G}$) is at least $C =\left \lceil (12\gamma m^3 T + 1)^{\ffrac{1}{\varepsilon}} \right \rceil$. However, when there is a valid $3$-partition, the envy output by $\mathsf{ALG}_{\cal G}$ is strictly upper bounded by:
\begin{align*}
    3m^2\gamma(CmT+1)^{1-\varepsilon} \le 6m^3 \gamma TC^{1-\varepsilon} \le 6m^3 \gamma T \left ( \left \lceil (12\gamma m^3 T + 1)^{\ffrac{1}{\varepsilon}} \right \rceil \right )^{1 - \varepsilon} < \left \lceil (12\gamma m^3 T + 1)^{\ffrac{1}{\varepsilon}} \right \rceil.
\end{align*}
This completes the proof.
\end{proof}

\subsection{General and Bounded-Degree Graphs}
In this section, we generalize the arguments from Section \ref{sec:treelowerbound} to other classes of graphs. 
The main technique is similar to that of Theorem \ref{thm:trees-approx-lower-bound}, so we just present ideas for the graph construction in each of these proofs, with the details in Appendix \ref{apdx:lower}.

We first match the $O(n^2)$ upper bound for connected graphs (\Cref{prop:trivialgeneral} and \Cref{cor:cutwidth-upperbounds}).

\begin{restatable}{theorem}{generalgraphs}\label{thm:general-approx-lower-bound}
For any constant $\varepsilon > 0$, there is no efficient $O(n^{2-\varepsilon})$ approximation algorithm for {\GHA} on connected graphs unless P = NP.
\end{restatable}
\begin{proof}[Proof Sketch]
We replace the $Ca_i$-sized stars in Figure \ref{fig:trees-reduction-tree} with $Ca_i$-sized cliques. The rest of the proof is similar to Theorem \ref{thm:trees-approx-lower-bound}.
\end{proof}



So far in our two lower bounds (Theorems \ref{thm:trees-approx-lower-bound} and \ref{thm:general-approx-lower-bound}), we were able to use simple counting techniques, because counting edges with non-zero envy in stars and cliques is straightforward. Our next results will require much more careful analysis.

We will start with bounded-degree planar graphs. Our reduction uses grid graphs instead of stars and cliques, and so we will need a technical lemma to help us with estimating the number of edges with nonzero envy.

\begin{restatable}{lemma}{gridlemma}\label{lem:grid-graph-property}
    Let $G = Grid(r, c)$ be a grid graph with $r$ rows and $c$ columns such that $r \le c$. Let $A \subseteq V$ be any set of nodes in this graph such that $|A| \le \ffrac{rc}{2}$. Then, $\delta_G(A) \geq \min\{\sqrt{|A|}, \ffrac{r}{2}\}$.
\end{restatable}
\begin{proof}
If $A$ consists of at least one node from each row, then since $|A| \le \ffrac{rc}{2}$, there will be at least $r/2$ rows with a node in $V \setminus A$. Therefore, there will be at least $\ffrac{r}{2}$ edges in the cut. Similarly, if $A$ consists of at least one node from each column, there will be $\ffrac{c}{2} \ge \ffrac{r}{2}$ edges in the cut. 

Otherwise, there is some row and some column containing only nodes in $V \setminus A$.
Note that there must either be at least $\sqrt{|A|}$ rows with a node in $A$ or at least $\sqrt{|A|}$ columns with a node in $A$. Assume WLOG there are at least $\sqrt{|A|}$ rows with a node in $A$. Each of these rows intersects the column that only has nodes from $V\setminus A$, and so each of the $\sqrt{|A|}$ rows must contain an edge between $A$ and $V\setminus A$.
\end{proof}

Armed with Lemma \ref{lem:grid-graph-property}, we can now present our lower bound on bounded-degree planar graphs.

\begin{restatable}{theorem}{bdp}\label{thm:bounded-degree-planar-approx-lower-bound}
For any constant $\varepsilon > 0$, no efficient $O(n^{0.5-\varepsilon})$ approximation algorithm exists for {\GHA} on bounded-degree planar graphs unless P = NP.
\end{restatable}
\begin{proof}[Proof Sketch]
We replace the stars of size $Ca_i$ in Figure \ref{fig:trees-reduction-tree} with grid graphs containing $C$ rows and $Ca_i$ columns. The rest of the proof flows similarly to Theorem \ref{thm:trees-approx-lower-bound}. Lemma \ref{lem:grid-graph-property} helps in estimating the envy blow-up if there is no $3$-partition.
\end{proof}




Note that Theorem \ref{thm:bounded-degree-planar-approx-lower-bound} matches the $O(\sqrt{n})$ upper bound from \Cref{cor:cutwidth-upperbounds}.

Our next lower bound applies to arbitrary bounded-degree graphs and matches the $O(n)$ upper bound from Proposition \ref{prop:trivialgeneral} and Corollary \ref{cor:cutwidth-upperbounds}. In this reduction, we use the recent polynomial-time algorithm \citep{cohen2016ramanujan} to compute bipartite Ramanujan multigraphs for any even number $m$ of vertices, and any degree $d \ge 3$. At a high level, we replace the star gadgets from the proof of Theorem \ref{thm:trees-approx-lower-bound} with these Ramanujan graphs and use the expansion properties of Ramanujan graphs to prove a lemma similar to (and stronger than) \Cref{lem:grid-graph-property}. 




\begin{restatable}{theorem}{boundeddeg}\label{thm:bounded-degree-approx-lower-bound}
For any constant $\varepsilon > 0$, there is no efficient $O(n^{1-\varepsilon})$ approximation algorithm for {\GHA} on bounded-degree graphs unless P = NP.
\end{restatable}

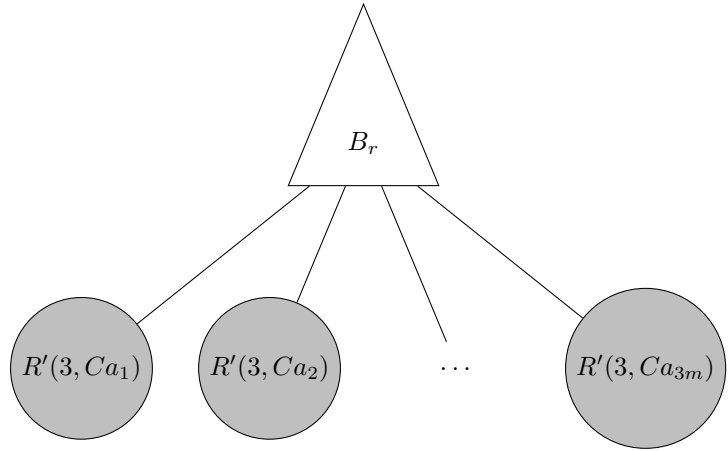
\begin{figure*}[ht]
    \centering
    \begin{tikzpicture}[,
    mycirc/.style={circle,fill=lightgray, draw = black,minimum size=0.75cm,inner sep = 3pt},
    mycirc2/.style={circle,fill=white,minimum size=0.75cm,inner sep = 3pt},
    mycirc3/.style={rectangle,fill=white,draw=black,minimum size=2cm,inner sep = 3pt},
    mycirc4/.style={isosceles triangle, shape border rotate = 90, fill=white,draw=black,minimum size=2cm,inner sep = 3pt},
    level 1/.style={sibling distance=25mm},
    level 2/.style={sibling distance=10mm},
    level distance = 3cm,
    BC/.style = {decorate,  
                     decoration={calligraphic brace, amplitude = 4mm, mirror},
                     very thick, pen colour={black}
                    }
    ]
    \node[mycirc4] {$B_r$} 
        child {node[mycirc] {$R'(3, Ca_1)$}}
        child {node[mycirc] {$R'(3, Ca_2)$}}
        child {node[mycirc2] {$\dots$}}
        child {node[mycirc] {$R'(3, Ca_{3m})$}};
    \end{tikzpicture}
    \caption{Mapping a \UTP instance to a bounded-degree graph.}
    \label{fig:bounded-degree-reduction-graph}
\end{figure*}

\begin{proof}[Proof Sketch]
    Figure \ref{fig:bounded-degree-reduction-graph} shows the graph we use for this reduction. The graph $G$ is constructed as follows. For each $a_i$ in the given {\UTP} instance, construct in polynomial time a $3$-regular Ramanujan bipartite multi-graph of size $Ca_i$ (using a result of \cite{cohen2016ramanujan}). Remove any repeated edges to convert them into simple graphs. The resulting graphs can be shown to have sufficient expansion properties, using Cheeger's inequality (Lemma \ref{lem:cheegers-inequality-apdx}). We can now attach these graphs to the leaves of a sufficiently small binary tree. This is a bounded-degree graph.

    If there is a $3$-partition, we can exhibit a small-envy allocation in the same way as in the other proofs in this section.
    Conversely, if there is no valid $3$-partition, we can show that some of these gadgets are going to be allocated multiple different values. We then use the expansion properties of these gadgets to show that the number of high envy edges within each of these gadgets is $\Omega(C)$.
\end{proof}




\subsection{Bounded-Degree Trees}
Our final lower bound shows that {\GHA} is NP-hard even when the underlying graph is a bounded degree tree. We still use \UTP in our reduction but this proof is significantly different from the previous ones. Our reduction will use a gadget we call the {\em flower}.\footnote{To the best of our knowledge, our specific flower graph is novel but it is possible (likely even) that the term ``flower'' has appeared before in the graph theory literature.}

\begin{definition}
The flower $F(n, k)$ is a rooted tree with $n$ nodes and maximum degree $k+1$, defined recursively as follows: for any $k \geq 1$, $F(1, k)$ is simply an isolated vertex which is the root node. For $n > 1$, $F(n, k)$ consists of a root node connected to the root nodes of $d$ other flowers $F(n_1, k), \dots, F(n_d, k)$ such that 
\begin{enumerate}[(a)]
    \item $\sum_{i = 1}^d n_i = n-1$, 
    \item if $n-1 \ge k$, then $d = k$ if $n$ and $k$ have different parities, and $d = k-1$ otherwise,
    \item each $n_i$ is odd, 
    \item for any $i, j \in [d]$, $|n_i - n_j| \le 2$.
\end{enumerate}
To ensure consistency with floral terminology, we refer to the root node of the flower $F(n, k)$ as its {\em pistil} and the (recursively smaller) flowers $F(n_1, k), \dots, F(n_d, k)$ as its {\em petals}.
\end{definition}

Before we use flowers, we show that they are well-defined and efficiently constructible. 

\begin{restatable}{lemma}{lemflowercomputation}\label{lem:flower-computation}
For any $n \ge 1$ and $k \ge 3$, the flower $F(n, k)$ exists and can be constructed in $\text{poly}(n, k)$ time. 
\end{restatable}
\begin{proof}
To prove existence, we only need to show that the numbers $n_1, \dots, n_d$ are guaranteed to exist. If $n -1 < k$, this is trivial: each $n_i = 1$.

Otherwise, if $n$ is even, then $d$ is required to be an odd value via condition (b). What we need to do is write $n-1$ as the sum of $d$ odd numbers. We set each $n_i$ as the greatest odd number which is at most $\lfloor \frac{n-1}{d} \rfloor$. Once we do this, $n-1 - \sum_{i = 1}^d n_i$ is guaranteed to be a non-negative even number which is strictly less than $2d$. So, we increment some of the $n_i$'s by $2$ till property (a) is satisfied. This construction satisfies the other two properties as well. 

The above argument not only shows that the numbers $n_1, \dots, n_d$ are guaranteed to exist but also presents a way to compute them in polynomial time. Since each $n_i$ is strictly less than $n$, we can use this subroutine to compute $F(n, k)$ recursively in polynomial time. More formally, let $T(n, k)$ be the complexity of constructing the flower $F(n, k)$. We have

\begin{align*}
    T(n, k) &= \left (\sum_{i = 1}^d T(n_i, k) \right ) + O(d) \\
    &\le k\cdot T((n-1)/(k-1) + 2, k) + O(k) \\
    &\le k\cdot T((n/k + 2, k) + O(k).
\end{align*}
If $n/k < 2$, then $T(n/k + 2, k)$ can be trivially constructed in $O(1)$ time since $n/k + 2 -1 \le 3 \le k$ and the graph is a constant-sized star. If $n/k \ge 2$, we can simplify the above expression as follows:
\begin{align*}
    T(n, k) \le k\cdot T(2n/k, k) + O(k),
\end{align*}
which simplifies to $T(n, k) = O(n^\frac{\ln 3}{\ln{3} - \ln2})$.
\end{proof}

The reason we build flowers is because they satisfy the two following useful properties.
\begin{restatable}{lemma}{lemflowerproperties}\label{lem:flower-properties}
Let $F(n, k)$ be a flower on the set of vertices $N$, and suppose $n \ge 10k$, and $n$ and $k$ have different parities. Then, $F(n, k)$ satisfies the following properties:
\begin{enumerate}[(i)]
    \item For any $A \subseteq N$ such that $|A|$ is even and $A$ does not contain the pistil, $\delta(A) \ge 2$.
    \item Each petal of $F(n, k)$ has size in the interval $\left [  \frac{4n}{5k}, \frac{6n}{5k} \right]$.
\end{enumerate}
\end{restatable}
\begin{proof}
(i) follows from property (c) in the definition of a flower. 
(ii) follows from the fact that there are $k$ petals (property (b)) and each petal has size in the interval $\left [ \frac{n}{k}-2, \frac{n}{k} + 2\right ]$ (property (d)).
\end{proof}

These simple properties are all we need to show the hardness of {\GHA} on bounded-degree trees.
\begin{restatable}{theorem}{thmboundeddegreetreesnp}\label{thm:bounded-degree-trees-np-complete}
{\GHA} is NP-hard on bounded-degree trees.   
\end{restatable}

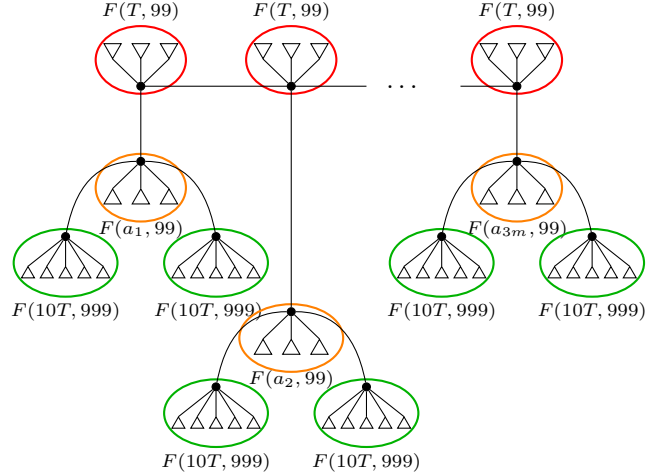
\begin{figure}[ht]
\centering
    \begin{tikzpicture}[xscale=0.5]
        \xdef\eps{0.05}
        %
        %
        \node[ellipse, draw=red, thick, minimum width=2cm, minimum height=1.5cm, scale=0.6] (e1) at (0,0.35) {};
        \node at (0,1) (l1) {\scriptsize{$F(T,99)$}};
        \node[ellipse, draw=red, thick, minimum width=2cm, minimum height=1.5cm, scale=0.6] (e2) at (4,0.35) {};
        \node at (4,1) (l2) {\scriptsize{$F(T,99)$}};
        \node[ellipse, draw=red, thick, minimum width=2cm, minimum height=1.5cm, scale=0.6] (e3) at (10,0.35) {};
        \node at (10,1) (l3) {\scriptsize{$F(T,99)$}};
        \node[ellipse, draw=orange, thick, minimum width=2cm, minimum height=1.5cm, scale=0.6] (e4) at (0,-1.35) {};
        \node at (0,-1.9) (l4) {\scriptsize{$F(a_1,99)$}};
        \node[ellipse, draw=orange, thick, minimum width=2cm, minimum height=1.5cm, scale=0.6] (e5) at (10,-1.35) {};
        \node at (10,-1.9) (l5) {\scriptsize{$F(a_{3m},99)$}};
        \node[ellipse, draw=black!30!green, thick, minimum width=2.3cm, minimum height=1.5cm, scale=0.6] (e6) at (-2,-2.35) {};
        \node at (-2,-3) (l6) {\scriptsize{$F(10T,999)$}};
        \node[ellipse, draw=black!30!green, thick, minimum width=2.3cm, minimum height=1.5cm, scale=0.6] (e7) at (2,-2.35) {};
        \node at (1.8,-3) (l7) {\scriptsize{$F(10T,999)$}};
        \node[ellipse, draw=black!30!green, thick, minimum width=2.3cm, minimum height=1.5cm, scale=0.6] (e8) at (8,-2.35) {};
        \node at (8,-3) (l8) {\scriptsize{$F(10T,999)$}};
        \node[ellipse, draw=black!30!green, thick, minimum width=2.3cm, minimum height=1.5cm, scale=0.6] (e9) at (12,-2.35) {};
        \node at (12,-3) (l9) {\scriptsize{$F(10T,999)$}};
        \node[ellipse, draw=orange, thick, minimum width=2.3cm, minimum height=1.5cm, scale=0.6] (e10) at (4,-3.35) {};
        \node at (4,-3.9) (l10) {\scriptsize{$F(a_2,99)$}};
        \node[ellipse, draw=black!30!green, thick, minimum width=2.3cm, minimum height=1.5cm, scale=0.6] (e11) at (2,-4.35) {};
        \node at (2,-5) (l11) {\scriptsize{$F(10T,999)$}};
        \node[ellipse, draw=black!30!green, thick, minimum width=2.3cm, minimum height=1.5cm, scale=0.6] (e12) at (6,-4.35) {};
        \node at (6,-5) (l12) {\scriptsize{$F(10T,999)$}};
        %
        %
        \node[circle,draw=black,fill=black,inner sep=0pt,minimum size=3pt] (1) at (0,0) {};
        \node[circle,draw=black,fill=black,inner sep=0pt,minimum size=3pt] (2) at (4,0) {};
        \node[circle,draw=black,fill=black,inner sep=0pt,minimum size=3pt] (3) at (10,0) {};
        %
        %
        \node[circle,draw=black,fill=black,inner sep=0pt,minimum size=3pt] (4) at (0,-1) {};
        \node[circle,draw=black,fill=black,inner sep=0pt,minimum size=3pt] (5) at (10,-1) {};
        %
        %
        \node[circle,draw=black,fill=black,inner sep=0pt,minimum size=3pt] (6) at (-2,-2) {};
        \node[circle,draw=black,fill=black,inner sep=0pt,minimum size=3pt] (7) at (2,-2) {};
        \node[circle,draw=black,fill=black,inner sep=0pt,minimum size=3pt] (8) at (8,-2) {};
        \node[circle,draw=black,fill=black,inner sep=0pt,minimum size=3pt] (9) at (12,-2) {};
        %
        %
        \node[circle,draw=black,fill=black,inner sep=0pt,minimum size=3pt] (10) at (4,-3) {};
        \node[circle,draw=black,fill=black,inner sep=0pt,minimum size=3pt] (11) at (2,-4) {};
        \node[circle,draw=black,fill=black,inner sep=0pt,minimum size=3pt] (12) at (6,-4) {};
        %
        %
        \draw (1) -- (4);
        \draw (2) -- (10);
        \draw (3) -- (5);
        \draw (4) to[out=180,in=70] (6);
        \draw (4) to[out=0,in=110] (7);
        \draw (10) to[out=180,in=70] (11);
        \draw (10) to[out=0,in=110] (12);
        \draw (5) to[out=180,in=70] (8);
        \draw (5) to[out=0,in=110] (9);
        %
        \node[regular polygon, draw, regular polygon sides=3,scale=0.4] (42) at (0,-1.5) {};
        \node[regular polygon, draw, regular polygon sides=3,scale=0.4] (41) at (-0.75,-1.5) {};
        \node[regular polygon, draw, regular polygon sides=3,scale=0.4] (43) at (0.75,-1.5) {};
        \draw (4) -- (41.north);
        \draw (4) -- (42.north);
        \draw (4) -- (43.north);
        %
        \node[regular polygon, draw, regular polygon sides=3,scale=0.4] (52) at (10,-1.5) {};
        \node[regular polygon, draw, regular polygon sides=3,scale=0.4] (51) at (9.25,-1.5) {};
        \node[regular polygon, draw, regular polygon sides=3,scale=0.4] (53) at (10.75,-1.5) {};
        \draw (5) -- (51.north);
        \draw (5) -- (52.north);
        \draw (5) -- (53.north);
        %
        \node[regular polygon, draw, regular polygon sides=3,scale=0.3] (63) at (-2,-2.5) {};
        \node[regular polygon, draw, regular polygon sides=3,scale=0.3] (61) at (-2.5,-2.5) {};
        \node[regular polygon, draw, regular polygon sides=3,scale=0.3] (62) at (-3,-2.5) {};
        \node[regular polygon, draw, regular polygon sides=3,scale=0.3] (64) at (-1.5,-2.5) {};
        \node[regular polygon, draw, regular polygon sides=3,scale=0.3] (65) at (-1,-2.5) {};
        \draw (6) -- (61.north);
        \draw (6) -- (62.north);
        \draw (6) -- (63.north);
        \draw (6) -- (64.north);
        \draw (6) -- (65.north);
        %
        \node[regular polygon, draw, regular polygon sides=3,scale=0.3] (73) at (2,-2.5) {};
        \node[regular polygon, draw, regular polygon sides=3,scale=0.3] (71) at (2.5,-2.5) {};
        \node[regular polygon, draw, regular polygon sides=3,scale=0.3] (72) at (3,-2.5) {};
        \node[regular polygon, draw, regular polygon sides=3,scale=0.3] (74) at (1.5,-2.5) {};
        \node[regular polygon, draw, regular polygon sides=3,scale=0.3] (75) at (1,-2.5) {};
        \draw (7) -- (71.north);
        \draw (7) -- (72.north);
        \draw (7) -- (73.north);
        \draw (7) -- (74.north);
        \draw (7) -- (75.north);
        %
        \node[regular polygon, draw, regular polygon sides=3,scale=0.3] (83) at (8,-2.5) {};
        \node[regular polygon, draw, regular polygon sides=3,scale=0.3] (81) at (8.5,-2.5) {};
        \node[regular polygon, draw, regular polygon sides=3,scale=0.3] (82) at (9,-2.5) {};
        \node[regular polygon, draw, regular polygon sides=3,scale=0.3] (84) at (7.5,-2.5) {};
        \node[regular polygon, draw, regular polygon sides=3,scale=0.3] (85) at (7,-2.5) {};
        \draw (8) -- (81.north);
        \draw (8) -- (82.north);
        \draw (8) -- (83.north);
        \draw (8) -- (84.north);
        \draw (8) -- (85.north);
        %
        \node[regular polygon, draw, regular polygon sides=3,scale=0.3] (93) at (12,-2.5) {};
        \node[regular polygon, draw, regular polygon sides=3,scale=0.3] (91) at (12.5,-2.5) {};
        \node[regular polygon, draw, regular polygon sides=3,scale=0.3] (92) at (13,-2.5) {};
        \node[regular polygon, draw, regular polygon sides=3,scale=0.3] (94) at (11.5,-2.5) {};
        \node[regular polygon, draw, regular polygon sides=3,scale=0.3] (95) at (11,-2.5) {};
        \draw (9) -- (91.north);
        \draw (9) -- (92.north);
        \draw (9) -- (93.north);
        \draw (9) -- (94.north);
        \draw (9) -- (95.north);
        %
        \node[regular polygon, draw, regular polygon sides=3,scale=0.4] (102) at (4,-3.5) {};
        \node[regular polygon, draw, regular polygon sides=3,scale=0.4] (101) at (3.25,-3.5) {};
        \node[regular polygon, draw, regular polygon sides=3,scale=0.4] (103) at (4.75,-3.5) {};
        \draw (10) -- (101.north);
        \draw (10) -- (102.north);
        \draw (10) -- (103.north);
        %
        \node[regular polygon, draw, regular polygon sides=3,scale=0.3] (113) at (2,-4.5) {};
        \node[regular polygon, draw, regular polygon sides=3,scale=0.3] (111) at (2.5,-4.5) {};
        \node[regular polygon, draw, regular polygon sides=3,scale=0.3] (112) at (3,-4.5) {};
        \node[regular polygon, draw, regular polygon sides=3,scale=0.3] (114) at (1.5,-4.5) {};
        \node[regular polygon, draw, regular polygon sides=3,scale=0.3] (115) at (1,-4.5) {};
        \draw (11) -- (111.north);
        \draw (11) -- (112.north);
        \draw (11) -- (113.north);
        \draw (11) -- (114.north);
        \draw (11) -- (115.north);
        %
        \node[regular polygon, draw, regular polygon sides=3,scale=0.3] (123) at (6,-4.5) {};
        \node[regular polygon, draw, regular polygon sides=3,scale=0.3] (121) at (6.5,-4.5) {};
        \node[regular polygon, draw, regular polygon sides=3,scale=0.3] (122) at (7,-4.5) {};
        \node[regular polygon, draw, regular polygon sides=3,scale=0.3] (124) at (5.5,-4.5) {};
        \node[regular polygon, draw, regular polygon sides=3,scale=0.3] (125) at (5,-4.5) {};
        \draw (12) -- (121.north);
        \draw (12) -- (122.north);
        \draw (12) -- (123.north);
        \draw (12) -- (124.north);
        \draw (12) -- (125.north);
        %
        \node[regular polygon, draw, regular polygon sides=3,scale=0.4, rotate=180] (012) at (0,0.5) {};
        \node[regular polygon, draw, regular polygon sides=3,scale=0.4, rotate=180] (011) at (0.75,0.5) {};
        \node[regular polygon, draw, regular polygon sides=3,scale=0.4, rotate=180] (013) at (-0.75,0.5) {};
        \draw (1) -- (011.north);
        \draw (1) -- (012.north);
        \draw (1) -- (013.north);
        %
        \node[regular polygon, draw, regular polygon sides=3,scale=0.4, rotate=180] (22) at (4,0.5) {};
        \node[regular polygon, draw, regular polygon sides=3,scale=0.4, rotate=180] (21) at (4.75,0.5) {};
        \node[regular polygon, draw, regular polygon sides=3,scale=0.4, rotate=180] (23) at (3.25,0.5) {};
        \draw (2) -- (21.north);
        \draw (2) -- (22.north);
        \draw (2) -- (23.north);
        %
        \node[regular polygon, draw, regular polygon sides=3,scale=0.4, rotate=180] (32) at (10,0.5) {};
        \node[regular polygon, draw, regular polygon sides=3,scale=0.4, rotate=180] (31) at (10.75,0.5) {};
        \node[regular polygon, draw, regular polygon sides=3,scale=0.4, rotate=180] (33) at (9.25,0.5) {};
        \draw (3) -- (31.north);
        \draw (3) -- (32.north);
        \draw (3) -- (33.north);
        \draw (0,0) -- (6,0);
        \draw (8.5,0) -- (10,0);
        \node at (7,0) (1111) {$\dots$};
\end{tikzpicture}
\caption{Mapping a \UTP instance to a bounded degree tree. Here, the orange, red, and green circles correspond to small, medium, and large flowers respectively.}
\label{fig:boundedtrees-reduction-graph}
\end{figure}



\begin{proof}
We present a reduction to {\UTP}. We assume the input $3$-partition instance is scaled up such that each $a_i$ is even and at least $1000$. This comes with no loss of generality --- we can multiply each $a_i$ and $T$ by a $1000$ without changing the output of the instance. We also assume $T$ is even. If $T$ is odd, the instance is trivially a NO instance.

We construct the graph of the {\GHA} instance as follows (see Figure \ref{fig:boundedtrees-reduction-graph}). We start with a path of $3m$ flowers $F(T, 99)$, connected by their pistils. We number these flowers $1$ to $3m$ from left to right. To each flower $i$, we connect a flower $F(a_i, 99)$ and to this flower, we connect two flowers $F(10T, 999)$. Again, connections between flowers are made via an edge between the pistils of the flowers. We refer to the flowers of the form $F(a_i, 99)$ as {\em small} flowers, flowers of the form $F(T, 99)$ as {\em medium} flowers and flowers of the form $F(10T, 999)$ as {\em large} flowers. We define small, medium and large pistils similarly. Note that large flowers have $999$ petals, while small and medium flowers have $99$ petals (since each $a_i$ and $T$ can be assumed to be even). This graph can be constructed in polynomial time using Lemma \ref{lem:flower-computation} and the fact that the inputs are given in unary.

Now we describe the house values. We need to describe $64mT$ house values since this is the number of nodes in the graph. For this, it is easier to first define the following function/series, $s(j) = (|E|+1)^{2j}$ for any $j \in \mathbb{N}$ where $|E|$ is the number of edges in our constructed graph. Note that equivalently, we can write $s(j) = (64mT)^{2j}$. We define the multiset of house values as the following:
\begin{align*}
    T &\text{ houses with value } 0 \\
    T &\text{ houses with value } s(4m) \\
    T &\text{ houses with value } s(4m) + s(4m-1) \\
    \dots \\
    T &\text{ houses with value } \sum_{j = 2}^{4m} s(j) \\
    60mT &\text{ houses with value } \sum_{j = 1}^{4m} s(j)
\end{align*}
The crucial property these values satisfy is that the intervals between two values are exponentially decreasing in size, so much so that minimum envy allocation must lexicographically minimize the number of edges passing through these gaps. That is, to minimize envy, we need to first minimize the number of lines passing through the $(0, s(4m))$ interval, subject to that minimize the number of lines passing through the $(s(4m), s(4m) + s(4m-1))$ interval, and so on. For ease of readability, we no longer refer to the exact value of the houses but simply refer to them as {\em clusters} of values. More specifically, the $T$ instances of $0$ are the \emph{first} cluster of values, the $T$ instances of $s(4m)$ are the \emph{second} cluster, and so on. There are $4m+1$ clusters, each of which have size $T$ except for the highest value (the $(4m+1)$-th cluster) which has size $60mT$. The largest house value is at most $(4m + 1)\cdot s(4m)$, which requires $O(\log(4m + 1) + (4m)\log(|E| + 1))$ bits to write, and there are still only polynomially many such values, and so this multiset of house values can be written in polynomial time and space.

Given this constructed instance, let us study the optimal allocation $\pi^{*}$. Note that the optimal allocation must first minimize the number of envy lines between the first two clusters. Therefore, the placement of the first cluster must be done in such a way so as to minimize the cut size between the first cluster and the rest of the graph. We will show that, to minimize envy, the first $3m$ clusters must be allocated to the $3m$ medium flowers, and the next $m$ clusters (from the $(3m+1)$-th to the $(4m)$-th cluster) must each be allocated to exactly three small flowers with sizes $a_i, a_j$ and $a_k$ such that $a_i + a_j + a_k = T$. This is possible if and only if the original instance is a YES instance, and the total envy of the corresponding house allocation is given by
\begin{equation*}
    \envy_{\text{YES}} = \sum_{j = 1}^{3m - 1}(j + 1)\cdot s(4m+1-j) + (3m)\cdot s(m+1)  + \sum_{j = 1}^m(3m + 3j)\cdot s(m+1-j) .
\end{equation*}
Any other allocation that does not follow this structure has a strictly greater envy. Specifically, this means that when the original instance is a NO instance, the optimal allocation has a greater total envy than $\envy_{\text{YES}}$ defined above.

To do this, we will need some bounds on petal sizes. The size of each petal in the medium flowers is at most $\frac{6T}{99 \times 5} \le \frac{T}{80}$ (Lemma \ref{lem:flower-properties}). Similarly the size of each petal in the large flowers is $\le \frac{T}{80}$ and the size of each petal in a small flower is $\frac{6a_i}{99 \times 5} \le \frac{6T}{2 \times 99 \times 5} \le \frac{T}{160}$. In particular, all petals throughout the graph have size at most $T/80$, a fact that we shall use several times.

Let us now consider where the first cluster is allocated in any optimal allocation. Note that in an allocation where the first cluster completely fills up one of the medium flowers at either end of the path (either leftmost or rightmost), the cut size across the first interval is $2$. Any optimal allocation therefore needs to ensure that this cut size is at most $2$. This observation will enable us to rule out many other possibilities.

If the elements of the first cluster are not allocated to any pistil, then the cut size is at least $80$ since each petal in the graph has size at most $\frac{T}{80}$; so the first cluster must be allocated to at least $80$ petals, each of which must add at least one edge to the cut. So in an optimal allocation, at least one pistil must get an element from the first cluster.

Suppose a large pistil is given a value in the first cluster. Then, since each petal of the large flower has size at least $T/125$, there must be at least $999-125 \ge 800$ petals of the large flower that remain unfilled by the cluster of values. Each of these unfilled petals adds at least one edge between the first cluster to another cluster, and these edges are all disjoint, and so the cut size must be at least $800$. Any such allocation, therefore, is suboptimal.

Suppose a small pistil is given a value in the first cluster. This pistil has two neighboring large pistils. We know from the arguments before that these neighbors cannot be in the first cluster, and so these two outgoing edges are across the cut. But then, the entire graph without these two large flowers is still connected, so at least one more edge needs to go across the first cut, and therefore, this allocation is also suboptimal.

It follows that any optimal allocation must allocate only medium pistils to the first cluster. If multiple such pistils are allocated, say from medium flowers $F_1$ and $F_2$, note that the cluster cannot contain either flower in its entirety, so there is at least one edge from each of them across the cut; but there also must be at least one other edge, from a path that goes from any of these two pistils to any node in the graph placed in a different cluster (such a node must exist, just by counting). This is therefore suboptimal. Similarly, if the first cluster contains exactly one pistil from a medium flower $F_1$, and no other pistil, but it does not contain all of $F_1$, it is easy to see that the cut will have size at least $3$, and so will be suboptimal.


So, assume from here on out that the first cluster is allocated entirely to the leftmost (or the first) medium flower under $\pi^*$. Our goal is to show that the first $3m$ clusters must be allocated to the $3m$ medium flowers from left to right. We do this by induction.

Assume the first $k$ clusters are allocated to the first $k$ medium flowers (from the left). Note that the cut size or the total number of lines of envy between the first $k$ clusters and the rest of the graph is $k+1$. Out of all the ways of allocating the rest of the items, we wish to find the one that minimizes the number of lines of envy going through the next interval, between the $(k + 1)$-th cluster and the $(k + 2)$-th one. We will show that in order to achieve this, the $(k+1)$-th cluster must be allocated to the $(k+1)$-th medium flower. If we allocate the $(k+1)$-th cluster to the $(k+1)$-th medium flower, the cut size between the first $k+1$ clusters and the rest of the graph increases by $1$ with respect to the first $k$ clusters i.e. it increases from $k+1$ to $k+2$. The increase in cut size is just an easier way to account for the number of edges between the first $k+1$ clusters and the rest of the graph. We show that all other allocations of the $(k+1)$-th cluster are strictly suboptimal i.e. all other allocations increase the cut size by at least $2$.

If the $(k+1)$-th cluster is not allocated to any pistil, then the cut size increases by at least $80$. If the $(k+1)$-th cluster is allocated to a large pistil, then the cut size increases by at least $800$. Both of these statements can be proved using similar arguments to the first cluster. 

For every small pistil that gets a value from the $(k+1)$-th cluster (regardless of whether it is attached to a medium flower that has already been assigned a cluster or not), the cut size increases by at least $1$ --- this is because each of these pistils has $2$ edges to large pistils. So even if the edge connecting the small pistil to the medium pistil is removed from the cut, at least two new edges are added. This argument rules out allocating the $(k+1)$-th cluster to multiple small flowers. If, on the other hand, no medium pistil is allocated and only one small pistil is allocated, at least $T/2$ values from the cluster must be allocated to petals of flowers whose pistil remains unallocated; this comes from the fact that each small flower has size upper bounded by $T/2$, and so the remaining $T/2$ values must be allocated to petals. Since every petal of any kind of flower has size at most $T/80$, we need at least $40$ different petals to be represented among the remaining values, and each of them has to add a distinct edge across the cut, and so the cut size must increase by at least $40$.

The only cases we have left are ones which involve at least one medium pistil. If multiple of these medium pistils are allocated, then since each of their petals has size at least $T/125$, there can be at most $125$ such petals represented in the cluster, and so we must have at least $198 - 125 \ge 50$ unfinished petals, each of which adds a distinct edge across the cut, increasing the cut size by at least $1$ each.

So exactly one medium pistil must be allocated to this cluster. We could still have a combination of one medium pistil and at least one small pistil, but this increases the cut size by at least $2$ since each of the small pistils increase the cut size by at least $1$ and the medium flower must be unfinished which adds at least another edge to the cut. 

Finally, we have the case where exactly one medium pistil and no other pistil is allocated. It is easy to see, using arguments similar to before, that the strictly optimal place to put this medium pistil is the pistil corresponding to the $(k+1)$-th medium flower and the strictly optimal allocation is to fill up the cluster entirely with the $(k+1)$-th medium flower.

We can conclude that the first $3m$ clusters must be allocated entirely to the medium flowers in order from left to right (or right to left). 
Note that we have not said anything yet to diffentiate instances with a valid $3$-partition. We will do that next.

The $(3m +1)$-th cluster still must be allocated to some pistil, but cannot be allocated to a large pistil, by the arguments from before. Note that all the medium flowers have already been allocated values from the first $3m$ clusters. Therefore, the $(3m + 1)$-th cluster must contain at least one pistil, and all pistils in it must be small. If the $(3m+1)$-th cluster fills up exactly $3$ different small flowers, then the cut size increases by exactly $3$ relative to the first $3m$ clusters. This is because $6$ edges from the small pistils to the large pistils are added, but the $3$ edges from the medium pistils to the small pistils are taken away. We will show that this is the best option, and any other allocation is strictly worse. To show this, we must consider the following four possible cases.

\noindent\textbf{Case 1: The $(3m+1)$-th cluster is allocated to exactly $1$ pistil.} This must be the pistil of a small flower, which we know has at most $T/2$ vertices by assumption. Therefore, at least $T/2$ values in the cluster must be allocated to petals of other flowers. Each of these petals adds at least one edge to the cut since their corresponding pistils have not been allocated values in the first $3m+1$ clusters, and these edges are all distinct. We know each petal in the graph has a size of at most $T/80$, so at least $40$ petals must be represented in this cluster, so the cut size increases by at least $40$, which is strictly suboptimal.

\noindent\textbf{Case 2: The $(3m+1)$-th cluster is allocated to exactly $2$ pistils.} Both these pistils must correspond to small flowers (say $F_1$ and $F_2$). Both these flowers have a combined size of strictly less than $T$, so at least one other flower outside these two is allocated a value from the $(3m+1)$-th cluster. If either $F_1$ or $F_2$ are not completely filled, then the cut size increases by at least $4$ --- two extra edges come from the pistils of $F_1$ and $F_2$ to the large pistils, the third is from at least one of $F_1$ or $F_2$ being unfinished, and the fourth is from the third flower which is allocated a value (note that it cannot be completely contained in this cluster, because its pistil is in a different cluster).

If both $F_1$ and $F_2$ are completely filled, then the argument is a little more subtle. Recall that each small pistil adds a $1$ to the cut size from the edges to the large pistils. We need to show that at least $2$ other edges are added to the cut from the values allocated outside $F_1$ and $F_2$. The number of values from the $(3m+1)$-th cluster allocated outside of $F_1$ and $F_2$ is {\em even}, since both $F_1$ and $F_2$ have even size by assumption. (Also, this number of values is nonzero, as $F_1$ and $F_2$ have a combined size of less than $T$ by assumption.) Therefore, if these values are allocated to two different flowers, at least $2$ more edges are added to the cut (since those two flowers have their pistils in other clusters) and we are done. Otherwise, if they are allocated to the same flower, the cut size still increases by at least $2$ because of Lemma \ref{lem:flower-properties} and we are done.

\noindent\textbf{Case 3: The $(3m+1)$-th cluster is allocated to exactly $3$ pistils but at least one of the flowers is incomplete.}
Here the cut size increase is trivially at least $4$ --- an increase of $3$ from the pistils having edges to the large pistils, and a fourth from the fact that at least one of the flowers is incomplete.

\noindent\textbf{Case 4: The $(3m+1)$-th cluster is allocated to $4$ or more pistils.}
It is a straightforward argument to show that the number of edges added across the cut is $4$ or more in this case.

Similarly, for the $(3m+2)$-th to the $(4m)$-th cluster, the best possible option is to allocate each cluster in a way that fills up $3$ small flowers. If any of these clusters fail to do so, the cut size is strictly higher. Given allocations of the first $4m$ clusters, there is only one possible allocation for the $(4m+1)$-th cluster; that is, the allocation where values from the $(4m+1)$-th cluster to all nodes which have not been allocated a value from the first $4m$ clusters.

It is easy to see that it is possible for all the $(3m+1)$-th to $(4m)$-th clusters to fill up three small flowers each if and only if there is a valid $3$-partition in the original instance. More specifically, if each of the $(3m+1)$-th to $(4m)$-th clusters are allocated in a way that attains a lower bound on the number of edges between the $(3m+k)$-th and the $(3m+k+1)$-th interval, this allocation defines a valid $3$-partition. This is because each of the $(3m+1)$-th to $(4m)$-th cluster must be allocated to exactly three small flowers and must fill them up entirely. It is also easy to see that given a valid $3$-partition, we can construct an allocation that attains the ideal total envy.

To compute the threshold, we proceed as follows. For $1 \leq i \leq 4m$, let $\ell_i$ be the length of the valuation interval between the $i$-th cluster and the $(i + 1)$-th cluster, which is precisely $s(4m + 1 - i)$. If the original {\UTP} instance is a YES instance, then by our arguments above, we can incur an envy of at most
\begin{equation*}
    \envy_{\text{YES}} := \sum_{j = 1}^{3m - 1}(j + 1)\cdot \ell_j + (3m)\cdot \ell_{3m} + \sum_{j = 1}^m(3m + 3j)\cdot\ell_{3m + j}.
\end{equation*}
Here, the interval $\ell_{3m}$ has increased the cut size by $0$, because it represents the rightmost medium flower being placed in the $3m$-th cluster. Note that the expression above can be computed in polynomial time, since each $\ell_j = s(4m + 1 - i)$ is representable in $O((4m + 1 - i)\log(|E|))$ bits, and all arithmetic operations can be done in polynomial space and time in this number of bits. There are only polynomially many operations to do, so we can compute the threshold $\envy_{\text{YES}}$ in polynomial time.

From all our analysis above, if the original {\UTP} instance is a YES instance, we can attain an envy of at most $\envy_{\text{YES}}$ on our constructed instance. On the other hand, if the original {\UTP} instance is a NO instance, we know any allocation is lexicographically worse than this expression. By our choice of the function $s(\cdot)$, this envy must be strictly more than $\envy_{\text{YES}}$. It follows that our {\GHA} instance has an allocation with envy at most $\envy_{\text{YES}}$ if and only if the {\UTP} instance is a YES instance. Since the construction is done entirely in polynomial time, this finishes the proof.
\end{proof}

\section{The Curious Case of Complete Binary Trees}\label{sec:completebintrees}


In this section, we investigate {\GHA} on instances where the underlying graph is a complete binary tree $B_k$. Recall that such a tree has depth $k$, and $2^{k+1} - 1$ vertices in total, of which $2^k$ are leaves. All leaves, furthermore, are at the same depth.

In \citet[Theorem 4.11]{canon}, it was shown that for any binary tree (complete or otherwise), at least one optimal allocation satisfies the \emph{local median property}: the value at every internal node is the median among the values given to that node and its two children. The same authors surmised that, for any binary tree, at least one optimal allocation satisfies the stronger \emph{global median property}: for every internal node $v$, either its left subtree gets strictly lower-valued houses and its right subtree gets strictly higher-valued houses, or the other way round. Note that if true, this would lead to a straightforward recursive polynomial-time algorithm that would compute an optimal allocation on (nearly) balanced binary trees. 

We now give a refutation of this conjecture. We illustrate an instance on a complete binary tree of depth $3$, in which no optimal allocation satisfies the global median property. This is a quite surprising result that shows that the general problem on complete binary trees may be much harder than expected.

\begin{example}\label{ex:globalrefutation}
    Consider the instance $(B_3, H)$, where
    \begin{equation*}
        H = \{0,0,0,0,0,0,0,1,1,1, 2,3,3,3,3\}.
    \end{equation*}
    See Figure \ref{fig:global-min-arg}. The top shows the only allocation satisfying the global median property (up to re-ordering). The total non-negligible envy incurred by this assignment comes out of the thick red edges of the $B_3$, which incur a total envy of $6$. However, the bottom shows an allocation with an envy of $5$ (incurred by the thick red edges), showing that the global median is strictly sub-optimal.
\end{example}

    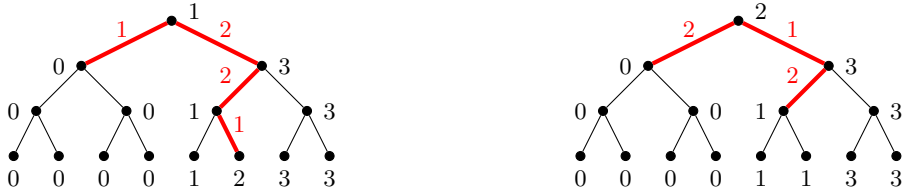
\begin{figure}[ht]
        \centering
        \small
        \begin{subfigure}[b]{0.45\textwidth}
            \centering
            \small
            \begin{tikzpicture}[,mycirc/.style={circle,fill=white, draw = black,minimum size=0.25cm,inner sep = 3pt},scale=0.6]
                \tikzset{mynode/.style = {shape=circle,fill=black,draw,inner sep=1.2pt}}
                \tikzset{edge/.style = {solid}}
                \tikzset{thickedge/.style = {solid,ultra thick,red}}
                \node[draw, mynode] at (0,0) (1) {};
                \node (1a) at (0,-0.5) {$0$};
                \node[draw, mynode] at (1,0) (2) {};
                \node (2a) at (1,-0.5) {$0$};
                \node[draw, mynode] at (2,0) (3) {};
                \node (3a) at (2,-0.5) {$0$};
                \node[draw, mynode] at (3,0) (4) {};
                \node (4a) at (3,-0.5) {$0$};
                \node[draw, mynode] at (4,0) (5) {};
                \node (5a) at (4,-0.5) {$1$};
                \node[draw, mynode] at (5,0) (6) {};
                \node (6a) at (5,-0.5) {$2$};
                \node[draw, mynode] at (6,0) (7) {};
                \node (7a) at (6,-0.5) {$3$};
                \node[draw, mynode] at (7,0) (8) {};
                \node (8a) at (7,-0.5) {$3$};
                \node[draw, mynode] at (0.5,1) (9) {};
                \node (9a) at (0,1) {$0$};
                \node[draw, mynode] at (2.5,1) (10) {};
                \node (10a) at (3,1) {$0$};
                \node[draw, mynode] at (4.5,1) (11) {};
                \node (11a) at (4,1) {$1$};
                \node[draw, mynode] at (6.5,1) (12) {};
                \node (12a) at (7,1) {$3$};
                \node[draw, mynode] at (1.5,2) (13) {};
                \node (13a) at (1,2) {$0$};
                \node[draw, mynode] at (5.5,2) (14) {};
                \node (14a) at (6,2) {$3$};
                \node[draw, mynode] at (3.5,3) (15) {};
                \node (15a) at (4,3.2) {$1$};
                \draw[edge] (9) -- (1);
                \draw[edge] (9) -- (2);
                \draw[edge] (10) -- (3);
                \draw[edge] (10) -- (4);
                \draw[edge] (11) -- (5);
                \draw[thickedge] (11) -- (6);
                \node (116) at (5,0.7) {\color{red} $1$};
                \draw[edge] (12) -- (7);
                \draw[edge] (12) -- (8);
                \draw[edge] (13) -- (9);
                \draw[edge] (13) -- (10);
                \draw[thickedge] (14) -- (11);
                \node (1114) at (4.7,1.8) {\color{red} $2$};
                \draw[edge] (14) -- (12);
                \draw[thickedge] (15) -- (13);
                \node (1513) at (2.4,2.8) {\color{red} $1$};
                \draw[thickedge] (15) -- (14);
                \node (1514) at (4.7,2.8) {\color{red} $2$};
            \end{tikzpicture}
        \end{subfigure}
        \begin{subfigure}[b]{0.45\textwidth}
            \centering
            \small
            \begin{tikzpicture}[,mycirc/.style={circle,fill=white, draw = black,minimum size=0.25cm,inner sep = 3pt},scale=0.6]
                \tikzset{mynode/.style = {shape=circle,fill=black,draw,inner sep=1.2pt}}
                \tikzset{edge/.style = {solid}}
                \tikzset{thickedge/.style = {solid,ultra thick, red}}
                \node[draw, mynode] at (0,0) (1) {};
                \node (1a) at (0,-0.5) {$0$};
                \node[draw, mynode] at (1,0) (2) {};
                \node (2a) at (1,-0.5) {$0$};
                \node[draw, mynode] at (2,0) (3) {};
                \node (3a) at (2,-0.5) {$0$};
                \node[draw, mynode] at (3,0) (4) {};
                \node (4a) at (3,-0.5) {$0$};
                \node[draw, mynode] at (4,0) (5) {};
                \node (5a) at (4,-0.5) {$1$};
                \node[draw, mynode] at (5,0) (6) {};
                \node (6a) at (5,-0.5) {$1$};
                \node[draw, mynode] at (6,0) (7) {};
                \node (7a) at (6,-0.5) {$3$};
                \node[draw, mynode] at (7,0) (8) {};
                \node (8a) at (7,-0.5) {$3$};
                \node[draw, mynode] at (0.5,1) (9) {};
                \node (9a) at (0,1) {$0$};
                \node[draw, mynode] at (2.5,1) (10) {};
                \node (10a) at (3,1) {$0$};
                \node[draw, mynode] at (4.5,1) (11) {};
                \node (11a) at (4,1) {$1$};
                \node[draw, mynode] at (6.5,1) (12) {};
                \node (12a) at (7,1) {$3$};
                \node[draw, mynode] at (1.5,2) (13) {};
                \node (13a) at (1,2) {$0$};
                \node[draw, mynode] at (5.5,2) (14) {};
                \node (14a) at (6,2) {$3$};
                \node[draw, mynode] at (3.5,3) (15) {};
                \node (15a) at (4,3.2) {$2$};
                \draw[edge] (9) -- (1);
                \draw[edge] (9) -- (2);
                \draw[edge] (10) -- (3);
                \draw[edge] (10) -- (4);
                \draw[edge] (11) -- (5);
                \draw[edge] (11) -- (6);
                \draw[edge] (12) -- (7);
                \draw[edge] (12) -- (8);
                \draw[edge] (13) -- (9);
                \draw[edge] (13) -- (10);
                \draw[thickedge] (14) -- (11);
                \node (1114) at (4.7,1.8) {\color{red} $2$};
                \draw[edge] (14) -- (12);
                \draw[thickedge] (15) -- (13);
                \node (1513) at (2.4,2.8) {\color{red} $2$};
                \draw[thickedge] (15) -- (14);
                \node (1514) at (4.7,2.8) {\color{red} $1$};
            \end{tikzpicture}
        \end{subfigure}
        \caption{Refutation of the global median property on complete binary trees.}
        \label{fig:global-min-arg}
    \end{figure}


Fix an arbitrary instance of {\GHA} on the complete binary tree $B_k$ on $n = 2^{k+1} - 1$ vertices, and consider the valuation interval. There are $n$ values on the interval. Of particular interest to us is the size of the \emph{smallest} $(i, n - i)$-cut, i.e., $\delta_{B_k}(i)$. Since $\delta_{B_k}(i) = \delta_{B_k}(n - i)$, we can WLOG take $i \leq \lceil n/2\rceil$. We now need a definition.

\begin{definition}[Repunit Representation and Elegance]\label{def:repunitrepresentation}
    For any $m \geq 1$, let a \emph{repunit representation of $m$} be any finite sequence $(a_1, \ldots, a_r) \in \mathbb{Z}^r$ satisfying
    \begin{equation*}
        m = \sum_{i =1}^r {\mathsf{sgn}(a_i)} \cdot (2^{|a_i|} - 1)
    \end{equation*}
    where $\mathsf{sgn}(a_i)$ is $1$ (resp.~$-1$) if $a_i \geq 0$ (resp.~$a_i < 0$). Note that every $m \geq 1$ has a repunit representation (e.g.,~the length-$m$ sequence of all ones). We define $\mathsf{elegance}(m)$ as the smallest $r$ for which $m$ has a repunit representation $(a_1, \ldots, a_r)$ of length $r$.
\end{definition}

The intuition behind Definition \ref{def:repunitrepresentation} is to capture the most ``efficient'' way to write $m$ in binary as the sum or difference of binary repunits, i.e., numbers of the form $11\ldots 1$. For instance, $\mathsf{elegance}(10) = 2$, because $10 = (2^3 - 1) + (2^2 - 1)$, and there is no shorter repunit representation. Similarly, $\mathsf{elegance}(12) = 2$, as $12 = (2^4 - 1) - (2^2 - 1)$. Note that $12$ cannot be written as the \emph{sum} of two repunits. Table \ref{tab:repunit} summarizes the elegance of all numbers up to $20$.

\begin{table*}[h!]
    \centering
    \begin{tabular}{ |c|c|c|c|c|c|c|c|c|c|c|c|c|c|c|c|c|c|c|c|c| } 
 \hline
 $m$ & 1 & 2 & 3 & 4 & 5 & 6 & 7 & 8 & 9 & 10 & 11 & 12 & 13 & 14 & 15 & 16 & 17 & 18 & 19 & 20 \\
 \hline
 $\mathsf{elegance}(m)$ & 1 & 2 & 1 & 2 & 3 & 2 & 1 & 2 & 3 & 2 & 3 & 2 & 3 & 2 & 1 & 2 & 3 & 2 & 3 & 4 \\
 \hline
\end{tabular}
    \caption{List of $\mathsf{elegance}(m)$ for $1 \leq m \leq 20$.}
    \label{tab:repunit}
\end{table*}


The following proposition relates elegance to the size of the smallest $(i,n-i)$-cut in a complete binary tree, namely $\delta_{B_k}(i)$.

\begin{restatable}{proposition}{elegancecuts}\label{prop:elegance}
    Let $B_k$ be the complete binary tree on $n = 2^{k + 1} - 1$ vertices. Then for $i \leq 2^k - 1$, $\mathsf{elegance}(i) - 1 \leq \delta_{B_k}(i) \leq \mathsf{elegance}(i)$.
\end{restatable}
\begin{proof}
    Consider \emph{any} $(i, n - i)$-cut in $B_k$, say $(S, V \setminus S)$, with $|S| = i$, and suppose there are $m'$ edges going across the cut. We will construct a repunit representation of $i$ with at most $m' + 1$ terms. Suppose the tree $B_k$ is rooted at the node $r$, and direct each edge from parent node to child node. Initialize a sequence $\vec{v}$ to be empty, and take any edge $e$ going across the cut $(S, V \setminus S)$. Observe that $e$ is either directed from $S$ to $V \setminus S$ or the other way round. This edge must have a complete binary subtree $B_{k_1}$ on one side in $B_k$: specifically, the subtree rooted at the child of $e$. If the edge $e$ is directed from $V \setminus S$ to $S$, then append the term $k_1 + 1$ to $\vec{v}$ (this corresponds to adding a repunit), and if $e$ is directed from $S$ to $V \setminus S$, append the term $-(k_1 + 1)$ (corresponding to subtracting a repunit). Once this is done for all edges $e$ going across the cut, we end up with a finite sequence $\vec{v}$. It is now easy to check that either this sequence $\vec{v}$, or the sequence $\vec{v}$ appended with the term $k + 1$, is a valid repunit representation of $i$. It follows that if the cut had been the minimum one, we would have a repunit representation of $i$ with at most $\delta_{B_k}(i) + 1$ terms. Therefore, $\mathsf{elegance}(i) \leq \delta_{B_k}(i) + 1$.

    Conversely, consider any optimal valid repunit representation of $i$ (which, therefore, has $\mathsf{elegance}(i)$ terms). Note that we can assume WLOG that all the terms are distinct in absolute value. This is because adding and subtracting the same term gives us a suboptimal representation, whereas $(2^{k_1} - 1) + (2^{k_1} - 1) = 2^{k_1 + 1} - 1 - 1$, so we can replace two additive repunits of the same length by two other unequal-length repunits without changing the result. We claim that we can assume WLOG that the largest repunit in this representation is at most $2^k - 1$. Otherwise, if it is $2^{k + 1} - 1$ or larger, then note that the next most significant repunit needs to be $-2^k + 1$, as otherwise the distinctness assumption gives us $i \geq (2^{k+1} - 1) - \sum_{j = 1}^{k-1}(2^j - 1) \geq 2^k$, contradiction. We can replace these two terms using $(2^{k+1} - 1) - (2^k - 1) = (2^k - 1) + 1$, which would replace these two terms by two other repunits without changing the value. This would be a valid representation, with all distinct terms unless the original representation also had a $+1$ term in it. We could then replace the $+1 + 1$ by $+3 - 1$, which would again be valid, unless the original had a $+3$ term in it. We could then replace the $+1 + 1 + 3$ by $+7 - 3 + 1$, which would again be valid, unless the original had a $+7$ term in it. We can keep going this way. What is the largest additive term in the original that we can run into in this way? Note that if we get to an additive term of $2^{k-1} - 1$, then we would have $i \geq (2^{k+1} - 1) - (2^k - 1) + (2^{k-1} - 1) - \sum_{j = 1}^{k-2}(2^j - 1) \geq 2^k$, contradiction. So the largest (additive) term can only be $2^{k-2} - 1$, and we would then terminate.
    
    So now we have an optimal repunit representation of $i$ with all distinct terms in absolute value, with the largest term being at most $2^k - 1$. We now claim that this repunit representation gives rise to an $(i, n - i)$ cut in $B_k$. We can just take complete binary subtrees of the sizes determined by the terms of the repunit representation, and include or exclude them on one side of the cut (according to the sign of the relevant term). The edges going across the cut will exactly be the edges to the roots of these subtrees, and the number of these edges will be exactly the number of terms in the original representation. This is \emph{some} $(i, n - i)$-cut of $B_k$, and so the smallest one has at most as many edges going across it as this one. It follows that $\delta_{B_k}(i) \leq \mathsf{elegance}(i)$.
\end{proof}

We note that if $i \ll n$, then in fact $\delta_{B_k}(i) = \mathsf{elegance}(i)$. Therefore, $\mathsf{elegance}(i)$ actually characterizes the size of the minimum $(i, n - i)$ cut in any sufficiently large binary tree.

Consider a value-agnostic algorithm for complete binary trees. Such an algorithm would need to assign the house values in any instance in some fixed order $(v_1, \ldots, v_n)$ to the vertices of $B_k$. 
The following proposition shows that doing this cannot simultaneously achieve the optimal cut on all smallest subintervals, and this leads to a lower bound on the approximability.

\begin{restatable}{proposition}{binarylowerbound}\label{prop:traversalimpossible}
    There is no value-agnostic algorithm for complete binary trees that attains an approximation better than $(5/3) \approx 1.67$.
\end{restatable}
\begin{proof}
    Take a large enough binary tree $B_k$ with $n \gg 100$ vertices, and consider the numbers 89 and 94. Note that $\mathsf{elegance}(89) = 3$, as $89 = 127 - 31 - 7$, and $\mathsf{elegance}(94) = 2$, as $94 = 63 + 31$. Furthermore, these are unique minimum-length repunit representations. We claim that no layout $\sigma = (v_1, \ldots, v_n)$ would attain $\delta_{B_k}(\{v_1, \ldots, v_{89}\}) = 3$ and $\delta_{B_k}(\{v_1, \ldots, v_{94}\}) = 2$ simultaneously. Indeed, if a layout $\sigma$ satisfies the first condition, then the root of a subtree with $127$ must receive one of the lowest $89$ values. However, if $\sigma$ also satisfies the second condition, then the $94$ lowest values fill up exactly two complete binary subtrees of size $63$ and $31$, and so the root of any subtree of size $127$ could not have have any of these values. This is a contradiction, and therefore, any value-agnostic algorithm for this complete binary tree needs to choose at most one of these two options. However, now consider two instances, $(B_k, H_1)$ and $(B_k, H_2)$, where $H_1$ consists of $89$ values of $0$ and $n - 89$ values of $1$, whereas $H_2$ consists of $94$ values of $0$ and $n - 94$ values of $1$. The optimal envy on $(B_k, H_1)$ is $\mathsf{elegance}(89) = 3$, whereas the optimal envy on $(B_k, H_2)$ is $\mathsf{elegance}(94) = 2$. A value-agnostic algorithm will yield a sub-optimal result on at least one of these two instances. If it is wrong on $(B_k, H_1)$, it has to have at least $5$ edges spanning the only nontrivial smallest subinterval (since it needs an odd number crossing the cut $\delta_{B_k}(\{v_1, \ldots, v_{89}\})$, as any repunit representation of $89$ needs an odd number of terms, by parity) and it will be off by a factor of at least $5/3 \approx 1.67$ on this instance. If it is wrong on $(B_k, H_2)$, it has to have at least $4$ edges crossing the only nontrivial smallest subinterval (since it needs an even number crossing the cut $\delta_{B_k}(\{v_1, \ldots, v_{94}\})$, again by parity) and it will be off by a factor of at least $4/2 = 2$ on this instance. Therefore, the approximation ratio has to be at least $1.67$.
\end{proof}

%
%
The counterexample in Proposition \ref{prop:traversalimpossible} and the failure of the global median property (Example \ref{ex:globalrefutation}) may seem to suggest that, even for complete binary trees, \emph{any} constant approximation ratio is unattainable. 
Remarkably, the following result shows that this is not the case: there is a \emph{value-agnostic} algorithm attaining a constant approximation on any complete binary tree. Indeed, ordering the vertices of $B_k$ in the standard in-order traversal and allocating the (sorted) values in that order yields a $3.5$-approximation.

\begin{restatable}{theorem}{inorder}\label{thm:inorder}
    Let $B_k$ be the complete binary tree on $n = 2^{k+1} - 1$ vertices. Then, on any house allocation instance on $B_k$, assigning the houses in increasing order to the vertices of $B_k$ in the standard in-order traversal gives us a total envy at most $3.5$ times the optimal value.
\end{restatable}
\begin{proof}
    Suppose we allocate the houses in sorted order to the vertices of $B_k$ in the standard in-order traversal. For any $i \leq 2^k - 1$, consider the number of edges of $B_k$ spanning the subinterval $(h_i, h_{i+1})$ of the valuation interval. It can be shown that under the in-order traversal, the number of edges spanning this interval is exactly $\mathsf{runs}(i)$, the number of runs of contiguous $0$s or $1$s in the binary representation of $i$, by a simple argument\footnote{ For instance, the two quantities follow the same recurrence relation: $f(2^k + i) = f(2^k - i + 1) + 1$ for $k \geq 0$ and $0 < i \leq 2^k$, with the same base cases.}.
    
    Our main claim will be to show that for all $i$, $\mathsf{runs}(i) \leq 3\cdot\mathsf{elegance}(i) - 2$. Let $\mathsf{elegance}(i) = r$, and consider an optimal repunit representation $(a_1, \ldots, a_r)$ of $i$. WLOG suppose $|a_1| \geq \ldots \geq |a_r|$. Then $\mathsf{sgn}(a_1) = 1$. We will start with the binary representation $11\ldots 1$ of $2^{a_1} - 1$, which contains a single run of exactly $a_1$ $1$s. We will then add or subtract all the other terms $a_2, \ldots, a_r$, performing all our operations in binary. We will carefully keep track of how each of these operations can affect the number of runs.

    Consider an arbitrary binary integer, with $t$ runs, and consider adding a repunit to it. Adding such a repunit can be thought of as adding a single power of $2$ (which is a binary integer of the form $10\ldots 0$), and then subtracting a single $1$. When we add the power of $2$, starting from the right, the $0$s do not change the number of runs, until we get to the leading $1$. Observe that adding or subtracting a single $1$ can increase the number of runs by at most $1$. Therefore, at the leading $1$, we can add a new run by a mismatched bit between the $0$ and the $1$, and can also add a new run by adding the $1$ itself. Therefore, adding a power of $2$ can increase the number of runs by at most $2$. After that, subtracting the $1$ adds at most another run, as observed. Therefore, adding a repunit adds at most three runs to the original binary integer. By a symmetric argument, subtracting a repunit (which is equivalent to subtracting a power of $2$, and then adding a $1$) can also increase the number of runs by at most $3$.

    Since we started with $2^{a_1} - 1$, which contained a single run, and then added or subtracted $r - 1$ other repunits, the total number of runs in the final integer is at most $1 + 3(r - 1)$. This immediately implies that $\mathsf{runs}(i) \leq 3\cdot\mathsf{elegance}(i) - 2$.

    Coming back to $B_k$, we have just shown that for $i \leq 2^k - 1$, the number of edges in the in-order traversal spanning the subinterval $(h_i, h_{i+1})$ is at most $3\cdot\mathsf{elegance}(i) - 2$, which by Proposition \ref{prop:elegance} is at most $3\cdot\delta_{B_k}(i) + 1$. We now consider a couple of cases.

    We note that, for $i \leq 2^k - 1$, we have $\delta_{B_k}(i) = 1$ if and only if $i = 2^{k_1} - 1$ for some $k_1 \leq k$. This follows just by observing that every edge in a complete binary tree has a complete binary subtree on one side, and the other side cannot have size less than $2^k$. But in that case, $\mathsf{runs}(i) = \delta_{B_k}(i)$, and so on these subintervals, the in-order traversal only subtends a single edge (and is therefore optimal).

    On the other hand, for $i \leq 2^k - 1$, if $\delta_{B_k}(i) \geq 2$, then $1 \leq (\ffrac{1}{2})\cdot\delta_{B_k}(i)$.

    Therefore, the number of edges over any subinterval $(h_i, h_{i+1})$ in the range $i \leq 2^k - 1$ is at most $3\cdot\delta_{B_k}(i) + (\ffrac{1}{2})\cdot\delta_{B_k}(i) = 3.5\cdot\delta_{B_k}(i)$. By symmetry, this is true over all smallest subintervals of the valuation interval. Therefore, the number of edges passing over every smallest subinterval is at most $3.5$ times the minimum possible number of edges passing over that subinterval, and this yields the desired result.
\end{proof}


\begin{remark}\label{rem:optimizedinorderproof}
    The proof of Theorem \ref{thm:inorder} can be optimized slightly for a more nuanced analysis. As observed in the proof, adding or subtracting repunits can be thought of as adding or subtracting powers of $2$, followed by subtracting or adding off $1$s. Consider performing all operations with the powers of $2$ first, and then finally adding or subtracting the number obtained by the $\pm 1$s. Each power-of-$2$ operation increases the number of runs by at most $2$, and the final additional number is at most $r - 1$, which has $1 + \lfloor\log(r - 1)\rfloor$ bits. Therefore, the total number of runs is actually $2 + 2(r - 1) + \lfloor\log(r - 1)\rfloor$. Note that this number is at most $(7/3)r$ for all $r \geq 18$, and in fact, the ratio of this number to $r$ gets arbitrarily close to $2$ as $r$ gets larger.
\end{remark}

It is instructive to check why this technique does not hold for arbitrary binary trees. Proposition \ref{prop:elegance} does not hold in general for non-complete binary trees. A complete binary tree ensures that there is always a binary subtree of the size given by a repunit representation to include on one side of the cut, but we lose this guarantee for non-complete trees.

We leave it as an open problem to construct either value-agnostic deterministic algorithms that achieve an approximation ratio better than $3.5$, or to obtain any polynomial-time algorithm (which cannot be value-agnostic) to obtain any approximation ratio better than $1.67$ for complete binary trees. We believe there should be an exact algorithm for this very special class of graphs, and hope that this will instigate future research into this problem.
\section{Conclusions}\label{sec:conclusions}

We explored the approximability of {\GHA}, presenting tight approximation algorithms for several classes of connected graphs, to our knowledge the first such results in the area. In particular, we gave polynomial-time algorithms exploiting graph structures to approximate the optimal envy on general graphs, trees, planar graphs, bounded-degree graphs, bounded-degree planar graphs, and bounded-degree trees; for each of these classes, we also gave a matching lower bound. Our algorithms were value-agnostic, i.e., they took into account only the input graph and the ordering among the house values but not the values themselves. We showed that any allocation on a random graph is a $(1 + o(1))$-approximation, and also gave a value-agnostic algorithm to show a $3.5$-approximation on all instances on complete binary trees.


The main question we leave for future work is the complexity of \GHA{} on complete binary trees. We know by the results in Section \ref{sec:completebintrees} that no exact algorithm can be value agnostic, but there seems to be no obvious way of leveraging the values, on even such a structured class of graphs.  
\begin{conjecture}\label{conj:completebintrees}
    {\GHA} is polynomial-time solvable on complete binary trees.
\end{conjecture}



\section*{Acknowledgments}

We wish to thank Justin Payan for many helpful discussions during the early stages of this work. We also thank Paul Seymour for some discussions related to the material in Section \ref{sec:completebintrees}. Rohit Vaish acknowledges support from SERB grant no. CRG/2022/002621 and DST INSPIRE grant no. DST/INSPIRE/04/2020/000107. Andrew McGregor and Rik Sengupta acknowledge support from NSF grant CCF-1934846. 
Hadi Hosseini acknowledges support from NSF IIS grants \#2144413 and \#2107173.

\bibliographystyle{plainnat}
\bibliography{abb,references}

\begin{thebibliography}{31}
\providecommand{\natexlab}[1]{#1}
\providecommand{\url}[1]{\texttt{#1}}
\expandafter\ifx\csname urlstyle\endcsname\relax
  \providecommand{\doi}[1]{doi: #1}\else
  \providecommand{\doi}{doi: \begingroup \urlstyle{rm}\Url}\fi

\bibitem[Aigner-Horev and Segal-Halevi(2022)]{aigner2022envy}
Elad Aigner-Horev and Erel Segal-Halevi.
\newblock {Envy-Free Matchings in Bipartite Graphs and their Applications to
  Fair Division}.
\newblock \emph{Information Sciences}, 587:\penalty0 164--187, 2022.

\bibitem[Alon and Spencer(2004)]{alon04probabilistic}
Noga Alon and Joel~H. Spencer.
\newblock \emph{The Probabilistic Method}.
\newblock Wiley, second edition, 2004.

\bibitem[Arora et~al.(1996)Arora, Frieze, and Kaplan]{frieze1996cutwidthdense}
Sanjeev Arora, Alan Frieze, and Haim Kaplan.
\newblock {A New Rounding Procedure for the Assignment Problem with
  Applications to Dense Graph Arrangement Problems}.
\newblock In \emph{Proceedings of 37th Conference on Foundations of Computer
  Science}, pages 21--30, 1996.

\bibitem[Beynier et~al.(2018)Beynier, Chevaleyre, Gourv\`{e}s, Lesca, Maudet,
  and Wilczynski]{beynier2018localenvy}
Aur\'{e}lie Beynier, Yann Chevaleyre, Laurent Gourv\`{e}s, Julien Lesca,
  Nicolas Maudet, and Ana\"{e}lle Wilczynski.
\newblock {Local Envy-Freeness in House Allocation Problems}.
\newblock In \emph{Proceedings of the 17th International Conference on
  Autonomous Agents and MultiAgent Systems}, page 292–300, 2018.

\bibitem[Bredereck et~al.(2022)Bredereck, Kaczmarczyk, and
  Niedermeier]{bredereck2022envy}
Robert Bredereck, Andrzej Kaczmarczyk, and Rolf Niedermeier.
\newblock {Envy-Free Allocations Respecting Social Networks}.
\newblock \emph{Artificial Intelligence}, 305:\penalty0 103664, 2022.

\bibitem[Bui and Jones(1992)]{bj92}
Thang~Nguyen Bui and Curt Jones.
\newblock {Finding Good Approximate Vertex and Edge Partitions is NP-Hard}.
\newblock \emph{Information Processing Letters}, 42\penalty0 (3):\penalty0
  153--159, 1992.

\bibitem[Chung(1984)]{mlatrees}
F.R.K. Chung.
\newblock {On Optimal Linear Arrangements of Trees}.
\newblock \emph{Computers \& Mathematics with Applications}, 10\penalty0
  (1):\penalty0 43--60, 1984.

\bibitem[Cohen(2016)]{cohen2016ramanujan}
Michael~B. Cohen.
\newblock {Ramanujan Graphs in Polynomial Time}.
\newblock In \emph{Proceedings of the 57th Symposium on Foundations of Computer
  Science (FOCS)}, pages 276--281, 2016.

\bibitem[Djidjev and Vrt'o(2006)]{djidjev2006cutwidth}
Hristo~N Djidjev and Imrich Vrt'o.
\newblock {Crossing Numbers and Cutwidths}.
\newblock \emph{Journal of Graph Algorithms and Applications. v7}, pages
  245--251, 2006.

\bibitem[Eiben et~al.(2020)Eiben, Ganian, Hamm, and
  Ordyniak]{eiben2020parameterized}
Eduard Eiben, Robert Ganian, Thekla Hamm, and Sebastian Ordyniak.
\newblock {Parameterized Complexity of Envy-Free Resource Allocation in Social
  Networks}.
\newblock In \emph{Proceedings of the Thirty-Fourth AAAI Conference on
  Artificial Intelligence}, pages 7135--7142, 2020.

\bibitem[Even et~al.(2000)Even, Naor, Rao, and Schieber]{even200mla}
Guy Even, Joseph~Seffi Naor, Satish Rao, and Baruch Schieber.
\newblock {Divide-and-Conquer Approximation Algorithms via Spreading Metrics}.
\newblock \emph{Journal of the ACM}, 47\penalty0 (4):\penalty0 585--616, 2000.

\bibitem[Even and Shiloach(1978)]{mlabinaryhard}
S.~Even and Y.~Shiloach.
\newblock {NP-Completeness of Several Arrangements Problems}.
\newblock \emph{Technical Report, TR-43 The Technicon}, page~29, 1978.

\bibitem[Feige and Lee(2007)]{feige2007mla}
Uriel Feige and James~R Lee.
\newblock {An Improved Approximation Ratio for the Minimum Linear Arrangement
  Problem}.
\newblock \emph{Information Processing Letters}, 101\penalty0 (1):\penalty0
  26--29, 2007.

\bibitem[Feldmann(2012)]{feldmann2012gridpartition}
Andreas~Emil Feldmann.
\newblock Fast balanced partitioning is hard even on grids and trees.
\newblock In \emph{Proceedings of the 37th International Symposium on
  Mathematical Foundations of Computer Science (MFCS)}, pages 372--382, 2012.

\bibitem[Feldmann and Foschini(2015)]{feldmann2015treepartition}
Andreas~Emil Feldmann and Luca Foschini.
\newblock Balanced partitions of trees and applications.
\newblock \emph{Algorithmica}, 71\penalty0 (2):\penalty0 354–376, feb 2015.

\bibitem[Gan et~al.(2019)Gan, Suksompong, and Voudouris]{gsvfairhouse}
Jiarui Gan, Warut Suksompong, and Alexandros~A. Voudouris.
\newblock {Envy-Freeness in House Allocation Problems}.
\newblock \emph{Mathematical Social Sciences}, 101:\penalty0 104--106, 2019.
\newblock ISSN 0165-4896.

\bibitem[Garey and Johnson(1979)]{garey1979computers}
M.~R. Garey and D.~S. Johnson.
\newblock \emph{Computers and Intractability: A Guide to the Theory of
  NP-Completeness}.
\newblock W. H. Freeman, first edition edition, 1979.

\bibitem[Garey et~al.(1976)Garey, Johnson, and Stockmeyer]{mlahard}
M.R. Garey, D.S. Johnson, and L.~Stockmeyer.
\newblock {Some Simplified NP-Complete Graph Problems}.
\newblock \emph{Theoretical Computer Science}, 1\penalty0 (3):\penalty0
  237--267, 1976.

\bibitem[Hosseini et~al.(2023)Hosseini, Payan, Sengupta, Vaish, and
  Viswanathan]{canon}
Hadi Hosseini, Justin Payan, Rik Sengupta, Rohit Vaish, and Vignesh
  Viswanathan.
\newblock {Graphical House Allocation}.
\newblock In \emph{Proceedings of the 2023 International Conference on
  Autonomous Agents and Multiagent Systems}, page 161–169, 2023.

\bibitem[Kamiyama(2021)]{kamiyama2021envy}
Naoyuki Kamiyama.
\newblock {The Envy-Free Matching Problem with Pairwise Preferences}.
\newblock \emph{Information Processing Letters}, 172:\penalty0 106158, 2021.

\bibitem[Kamiyama et~al.(2021)Kamiyama, Manurangsi, and
  Suksompong]{kmsfairhouse}
Naoyuki Kamiyama, Pasin Manurangsi, and Warut Suksompong.
\newblock {On the Complexity of Fair House Allocation}.
\newblock \emph{Operations Research Letters}, 49\penalty0 (4):\penalty0
  572--577, 2021.

\bibitem[Korach and Solel(1993)]{korach1993cutwidth}
Ephraim Korach and Nir Solel.
\newblock {Tree-Width, Path-Width, and Cutwidth}.
\newblock \emph{Discrete Applied Mathematics}, 43\penalty0 (1):\penalty0
  97--101, 1993.

\bibitem[Leighton and Rao(1999)]{leighton1999cutwidthapprox}
Tom Leighton and Satish Rao.
\newblock {Multicommodity Max-Flow Min-Cut Theorems and Their Use in Designing
  Approximation Algorithms}.
\newblock \emph{Journal of the ACM}, 46\penalty0 (6):\penalty0 787--832, 1999.

\bibitem[Lubotzky et~al.(1988)Lubotzky, Phillips, and
  Sarnak]{lubotzky1988ramanujan}
Alexander Lubotzky, Ralph Phillips, and Peter Sarnak.
\newblock {Ramanujan Graphs}.
\newblock \emph{Combinatorica}, 8\penalty0 (3):\penalty0 261--277, 1988.

\bibitem[Madathil et~al.(2023)Madathil, Misra, and Sethia]{MMS23complexity}
Jayakrishnan Madathil, Neeldhara Misra, and Aditi Sethia.
\newblock {The Complexity of Minimizing Envy in House Allocation}.
\newblock In \emph{Proceedings of the 2023 International Conference on
  Autonomous Agents and Multiagent Systems}, pages 2673--2675, 2023.

\bibitem[Monien and Sudborough(1988)]{monien1988cutwidth}
Burkhard Monien and Ivan~Hal Sudborough.
\newblock {Min Cut is NP-Complete for Edge Weighted Trees}.
\newblock \emph{Theoretical Computer Science}, 58\penalty0 (1-3):\penalty0
  209--229, 1988.

\bibitem[Rao and Richa(2005)]{richarao2005mla}
Satish Rao and Andr{\'e}a~W Richa.
\newblock {New Approximation Techniques for Some Linear Ordering Problems}.
\newblock \emph{SIAM Journal on Computing}, 34\penalty0 (2):\penalty0 388--404,
  2005.

\bibitem[Seidvasser(1970)]{seidvasser}
M.~A. Seidvasser.
\newblock {The Optimal Number of Vertices of a Tree}.
\newblock \emph{Diskref. Anal.}, 19:\penalty0 56--74, 1970.

\bibitem[Shapley and Scarf(1974)]{shapley1974cores}
Lloyd Shapley and Herbert Scarf.
\newblock {On Cores and Indivisibility}.
\newblock \emph{Journal of Mathematical Economics}, 1\penalty0 (1):\penalty0
  23--37, 1974.

\bibitem[Svensson(1999)]{svensson1999strategy}
Lars-Gunnar Svensson.
\newblock {Strategy-Proof Allocation of Indivisible Goods}.
\newblock \emph{Social Choice and Welfare}, 16\penalty0 (4):\penalty0 557--567,
  1999.

\bibitem[Yannakakis(1985)]{yannakakis1985treecutwidth}
Mihalis Yannakakis.
\newblock {A Polynomial Algorithm for the Min-Cut Linear Arrangement of Trees}.
\newblock \emph{Journal of the ACM}, 32\penalty0 (4):\penalty0 950--988, 1985.

\end{thebibliography}

\newpage

\appendix

\section{Proofs from Section \ref{sec:lower}}\label{apdx:lower}

\generalgraphs*
\begin{proof}
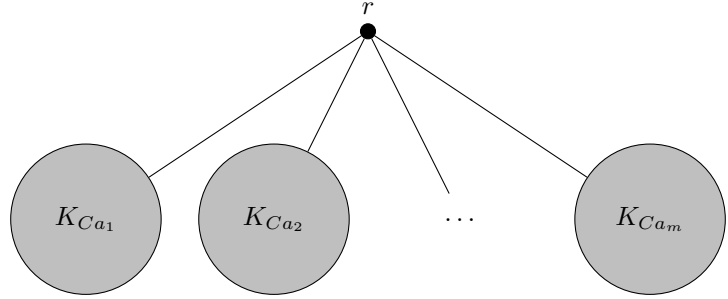
\begin{figure*}[ht]
    \centering
    \begin{tikzpicture}[,
    mycirc/.style={circle,fill=black, draw = black,minimum size=0.2cm,inner sep = 1.2pt},
    mycirc2/.style={circle,fill=white,minimum size=0.75cm,inner sep = 3pt},
    mycirc3/.style={circle,fill=lightgray,draw=black,minimum size=2cm,inner sep = 3pt},
    level 1/.style={sibling distance=25mm},
    level 2/.style={sibling distance=10mm},
    level distance = 2.5cm,
    BC/.style = {decorate,  
                     decoration={calligraphic brace, amplitude = 4mm, mirror},
                     very thick, pen colour={black}
                    }
    ]
    \node[mycirc,label=above:$r$] {} 
        child {node[mycirc3] {$K_{Ca_1}$}}
        child {node[mycirc3] {$K_{Ca_2}$}}
        child {node[mycirc2] {$\dots$}}
        child {node[mycirc3] {$K_{Ca_m}$}};
    \end{tikzpicture}
    \caption{Mapping a \UTP instance to a connected graph.}
    \label{fig:general-reduction-graph}
\end{figure*}
We present a similar reduction from the \UTP Problem. The main difference in construction from Theorem \ref{thm:trees-approx-lower-bound} is that $x_i$ and its children are replaced with a clique of size $Ca_i$. We use the high density of the clique to show that the envy is much higher $(\Omega(C^2))$ in the absence of a valid $3$-partition.  

For some constant $\varepsilon > 0$, assume there is an efficient $O(n^{2- \varepsilon})$ approximation algorithm $\mathsf{ALG}_{\cal G}$ where $\cal G$ is the class of all connected graphs. Assume there is some constant $\gamma$ such that for all instances on connected graphs, $\mathsf{ALG}_{\cal G}$ outputs an allocation with total envy within a multiplicative factor of $\gamma n^{2- \varepsilon}$ to the optimal envy.

Given an instance of \UTP, we construct an instance of {\GHA} as follows: for each value $a_i$ in the multiset $A$, we create a clique of size $Ca_i$. We then connect these $3m$ cliques to a node $r$ (as described in Figure \ref{fig:general-reduction-graph}). Once again, $C$ is a positive integer whose exact value we shall choose later.
The total number of nodes in this graph is $CmT+1$, similar to before. The house values are defined similarly as well: for each $j \in [m]$, there are $CT$ houses valued at $j$ and there is one house valued at $0$.

We show that, with the appropriate $C$, the total envy output by $\mathsf{ALG}_{\cal G}$ for the constructed {\GHA} instance is greater than or equal to $\left \lceil (96\gamma m^4 T^2 + 1)^{\ffrac{1}{\varepsilon}} \right \rceil^2$ if and only if there exists {\em no} valid solution for the original \textsc{3-Partition} instance, i.e.,~the original instance was a NO instance.

$(\Leftarrow)$ Assume there is a valid $3$-partition for the original instance. This case follows the exact same way as Theorem \ref{thm:trees-approx-lower-bound}. The minimum envy is upper bounded by $3m^2$ and the envy output by $\mathsf{ALG}_{\cal G}$ is upper bounded by $3m^2\gamma(CmT+1)^{2-\varepsilon}$.

$(\Rightarrow)$ Assume there is no valid $3$-partition in the original instance. We will show that any allocation must have an envy of at least $\ffrac{C^2}{4}$. 

If, for some allocation $\pi$, there exists a clique where no value $j \in [m] \cup \{0\}$ is allocated to more than half of its nodes, then this lower bound trivially holds. Since the clique has a size of at least $C$, each node in the clique envies at least $C/2$ neighbors by at least $1$. 

Assume that for all cliques, there is some value allocated to at least half the nodes in the clique; we refer to this value as a {\em majority value} of the clique. Let clique $i$ correspond to the clique $K_{Ca_i}$. For notational convenience, assume WLOG that cliques $\{1, \dots, \ell_1\}$ have majority value $1$, $\{{\ell_1 +1}, \dots, {\ell_1 + \ell_2}\}$ have majority value $2$ and so on. Using analysis similar to \Cref{thm:trees-approx-lower-bound}, there exists at least one $j \in [m]$ such that $\sum_{h \in [\ell_j]}a_{\ell_1 + \dots + \ell_{j-1} +h} > T$. Assume again for notational convenience that $j = 1$. Since all these values are integers, we can restate the equation above as $\sum_{h \in [\ell_1]}a_{h} \ge T+ 1$. This implies $\sum_{h \in [\ell_1]} Ca_{h} \ge CT + C$. 

Coming back to our allocation $\pi$, $\sum_{h \in [\ell_1]} Ca_{h} \ge CT + C$ implies that there are at least $C$ nodes in cliques $\{1,2, \dots, \ell_1\}$ which are not allocated a value of $1$. Since $1$ is a majority value in each of these cliques, the envy that each of the $C$ nodes in $S$ will have towards the nodes with value $1$ is at least $C/2$. This implies that the total envy of allocation $\pi$ is at least $C^2/2$. 

We set $C = 2\left \lceil (96\gamma m^4 T^2 + 1)^{\ffrac{1}{\varepsilon}} \right \rceil$ to complete the reduction. When there is no valid $3$-partition, the minimum total envy (and therefore, the envy output by $\mathsf{ALG}_{\cal G}$) is at least $\ffrac{C^2}{4} = \left \lceil (96\gamma m^4 T^2 + 1)^{\ffrac{1}{\varepsilon}} \right \rceil^2$. However, when there is a valid $3$-partition, the envy output by $\mathsf{ALG}_{\cal G}$ is strictly upper bounded by:
\begin{align*}
    3m^2\gamma(CmT+1)^{2-\varepsilon} \le 6m^4 \gamma T^2C^{2-\varepsilon} \le 24m^4 \gamma T^2 \left ( \left \lceil (96\gamma m^4 T^2 + 1)^{\ffrac{1}{\varepsilon}} \right \rceil \right )^{2 - \varepsilon} < \left \lceil (96\gamma m^4 T^2 + 1)^{\ffrac{1}{\varepsilon}} \right \rceil^2.
\end{align*}
This concludes the proof.
\end{proof}



\bdp*
\begin{proof}
\begin{figure*}[ht]
    \centering
    \begin{tikzpicture}[,
    mycirc/.style={circle,fill=white, draw = black,minimum size=0.75cm,inner sep = 3pt},
    mycirc2/.style={circle,fill=white,minimum size=0.75cm,inner sep = 3pt},
    mycirc3/.style={rectangle,fill=lightgray,draw=black,minimum size=2cm,inner sep = 3pt},
    mycirc4/.style={isosceles triangle, shape border rotate = 90, fill=white,draw=black,minimum size=2cm,inner sep = 3pt},
    level 1/.style={sibling distance=25mm},
    level 2/.style={sibling distance=10mm},
    level distance = 3cm,
    BC/.style = {decorate,  
                     decoration={calligraphic brace, amplitude = 4mm, mirror},
                     very thick, pen colour={black}
                    }
    ]
    \node[mycirc4] {$B_r$} 
        child {node[mycirc3] {$\Grid(C, Ca_1)$}}
        child {node[mycirc3] {$\Grid(C, Ca_2)$}}
        child {node[mycirc2] {$\dots$}}
        child {node[mycirc3] {$\Grid(C, Ca_{3m})$}};

    \end{tikzpicture}
    \caption{Mapping a \UTP instance to a bounded-degree planar graph.}
    \label{fig:bdp-reduction-graph}
\end{figure*}
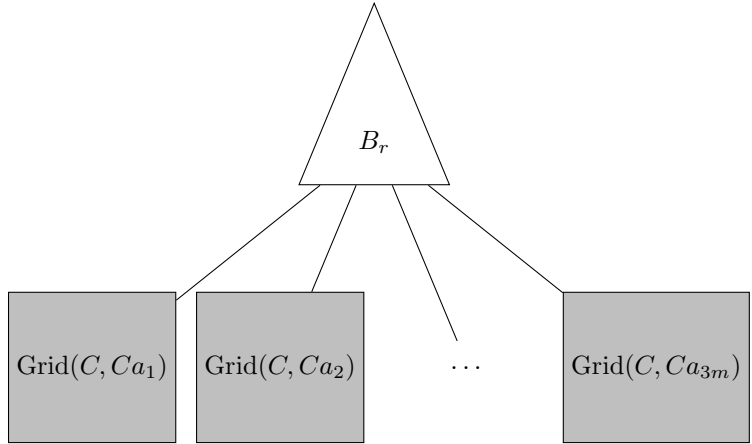
We present a similar reduction from the \UTP Problem. 
For some constant $\varepsilon > 0$, assume there is an efficient $O(n^{0.5 - \varepsilon})$ approximation algorithm $\mathsf{ALG}_{\cal G}$ where $\cal G$ corresponds to the class of all planar graphs with max degree at most $5$. In other words, for all instances on these graphs, $\mathsf{ALG}_{\cal G}$ outputs an allocation with total envy within a multiplicative factor of $\gamma n^{0.5- \varepsilon}$ to the optimal envy, where $\gamma$ is some fixed constant.

Given an instance of \UTP, we construct an instance of {\GHA} as follows: the graph $G$ has $3m$ grids, each grid $i \in [3m]$ has $C$ rows and $C a_i$ columns for some $C$ we will define later. Each of the grids has a single edge to a unique leaf of a binary tree $B_r$ (see Figure \ref{fig:bdp-reduction-graph}). 

The binary tree $B_r$ is constructed such that it has at most $12m$ nodes. It is easy to see that the smallest balanced binary tree with at least $3m$ leaves satisfies this constraint. This construction is trivially a planar graph with degree at most $5$. There are $C^2mT + |B_r|$ nodes in this graph. The house values are defined roughly the same way as before: for each $j \in [m]$, there are $C^2T$ houses valued at $j$ and there are $|B_r|$ houses valued at $0$. As long as $C$ and $|B_r|$ are polynomial, this is a polynomial-time reduction.

We show that, with an appropriate choice of $C$, the minimum total envy output by $\mathsf{ALG}_{\cal G}$ for the house allocation instance is greater than or equal to $\left \lceil (288\gamma m^3 T + 1)^{\ffrac{1}{(2\varepsilon)}} \right \rceil$ if and only if there exists {\em no} valid partition for the original \textsc{3-Partition} instance i.e. the original instance was a NO instance.

$(\Leftarrow)$ Assume there is a valid $3$-partition for the original instance. This case follows the exact same way as Theorem \ref{thm:trees-approx-lower-bound}. The minimum envy is upper bounded by $3m^2$ and the envy output by $\mathsf{ALG}_{\cal G}$ is upper bounded by $3m^2\gamma(C^2mT+|B_r|)^{0.5-\varepsilon} \le 3m^2\gamma(C^2mT+12m)^{0.5-\varepsilon}$.

$(\Rightarrow)$ Assume there is no valid $3$-partition in the original instance. We will show that any allocation must have an envy of at least $\ffrac{C}{4}$. 

Similar to Theorem \ref{thm:general-approx-lower-bound}, we refer to a value $j \in [m] \cup \{0\}$ as a majority value of a grid $\Grid(C, Ca_i)$ if this value has been allocated to at least $\ffrac{C^2a_i}{2}$ nodes in the grid. 

For some allocation $\pi$, assume there exists a grid, $\Grid(C, Ca_1)$ with no majority value. Then this grid has nodes with at least $3$ different values. Let $Q_j$ be the set of nodes with value $j$ in this grid under allocation $\pi$. Note that the number of edges from $Q_j$ to nodes with a different value is at least $\min\{\sqrt{|Q_j|}, \ffrac{C}{2}\}$ (Lemma \ref{lem:grid-graph-property}). Therefore the total number of edges between nodes with different value is lower bounded by 
\[    \frac{\sum_{j=0}^k \min\{\sqrt{|Q_j|}, \ffrac{C}{2}\}}{2} \ge
    \frac{\min \{\sum_{j=0}^k \sqrt{|Q_j|}, \ffrac{C}{2}\}}{2} 
    \ge \frac{\min\{ \sqrt{\sum_{j=0}^k |Q_j|}, \ffrac{C}{2}\}}{2} 
    \ge \ffrac{C}{4}.
\]
Note that we divide by $2$ since each edge gets counted at most twice. Since each of these edges has an envy of $1$, we are done.

From here on, assume that each grid has a majority value. Note that a grid can only have one majority value.
Let grid $i$ correspond to $\Grid(C,Ca_i)$. For notational convenience, assume WLOG that grids $\{1, \dots, \ell_1\}$ have majority value $1$, $\{{\ell_1 +1}, \dots, {\ell_1 + \ell_2}\}$ have majority value $2$ and so on. Using analysis similar to \Cref{thm:trees-approx-lower-bound}, there exists one $j \in [m]$ such that $\sum_{h \in [\ell_j]}a_{\ell_1 + \dots + \ell_{j-1} +h} > T$. Assume again for notational convenience that $j = 1$. Since all these values are integers, we can restate it as $\sum_{h \in [\ell_1]}a_{h} \ge T + 1$. This implies $\sum_{h \in [\ell_1]} C^2a_{h} \ge C^2T + C^2$. 

Coming to our allocation $\pi$, $\sum_{h \in [\ell_1]} C^2a_{h} \ge C^2T + C^2$ implies that there are at least $C^2$ nodes in grids $\{1,2, \dots, \ell_1\}$ which are not allocated a value of $1$. Let $A_i$ correspond to the set of nodes in grid $i$ allocated a value other than $1$. Note that $\sum_{i \in [\ell_1]} |A_i| \ge C^2$. We can lower bound the number of edges from $A_i$ to nodes with value $1$ using Lemma \ref{lem:grid-graph-property} as follows:
\[
    \sum_{i \in [\ell_1]} \min\{\sqrt{|A_i|}, \ffrac{C}{2}\} \ge
    \min\{\sum_{i \in [\ell_1]} \sqrt{|A_i|}, \ffrac{C}{2}\} 
    \ge \min\{ \sqrt{\sum_{i \in [\ell_1]}|A_i|}, \ffrac{C}{2}\} 
    \ge \ffrac{C}{2} \ .
\] 
Each of these edges have envy at least $1$.

We set $C = 4\left \lceil (288\gamma m^3 T + 1)^{\ffrac{1}{(2\varepsilon)}} \right \rceil$ to complete the reduction. When there is no valid $3$-partition, the total minimum envy (and therefore, the envy output by $\mathsf{ALG}_{\cal G}$) is at least $\ffrac{C}{4} = \left \lceil (288\gamma m^3 T + 1)^{\ffrac{1}{(2\varepsilon)}} \right \rceil^2$. However, when there is a valid $3$-partition, the envy output by $\mathsf{ALG}_{\cal G}$ is strictly upper bounded by:
\begin{align*}
    3m^2\gamma(C^2mT+12m)^{0.5-\varepsilon} \le 36m^3 \gamma TC^{1-2\varepsilon} <& 144m^3 \gamma T \left ( \left \lceil (288\gamma m^3 T + 1)^{\ffrac{1}{(2\varepsilon)}} \right \rceil \right )^{1 - 2\varepsilon} \\ <& \left \lceil (288\gamma m^3 T + 1)^{\ffrac{1}{(2\varepsilon)}} \right \rceil.\qedhere
\end{align*}
\end{proof}

\boundeddeg*
\begin{proof}

We, again, present a reduction from the \UTP Problem. 
For some constant $\varepsilon > 0$, assume there is an efficient $O(n^{1 - \varepsilon})$ approximation algorithm $\mathsf{ALG}_{\cal G}$ where $\cal G$ is the set of all connected graphs with degree at most $4$. For all instances, $\mathsf{ALG}_{\cal G}$ outputs an allocation with total envy within a multiplicative factor of $\gamma n^{1- \varepsilon}$ to the optimal envy for some constant $\gamma$.

Given an instance of \UTP, we construct $3m$ disjoint graphs as follows: for each $a_i$ in the multiset $A$, we construct a 3-regular Ramanujan bipartite multigraph of size $Ca_i$. This can be done in polynomial time and is well-defined when $Ca_i$ is an even integer greater than or equal to $6$ \citep{cohen2016ramanujan}. We will ensure this by setting $C$ appropriately. Note that these $3m$ multigraphs still may have multiple edges between the same pair of nodes. To convert them into simple graphs, we simply remove any repeated edges. For each $a_i \in A$, we refer to this graph as the $R'$-graph $i$ (or simply $R'(3, Ca_i)$). These graphs are neither Ramanujan graphs, nor are they $3$-regular. However, as we will crucially show, they still have the expansion properties we require. This result will use the Cheeger's inequality from \citet[Section 9.2]{alon04probabilistic} applied to Ramanujan graphs \citep{lubotzky1988ramanujan}. While the result in \citet[Section 9.2]{alon04probabilistic} is for simple graphs, the exact same proof can be extended to multigraphs; so we present it without proof.

\begin{lemma}[Cheeger's Inequality]\label{lem:cheegers-inequality-apdx}
Let $G'$ be a $d$-regular Ramanujan graph (or multigraph) defined on a set of $V$ nodes. Then,
\begin{align*}
    \min_{S \subseteq V: 0 < |S| \le |V|/2} \frac{\delta_{G'}(S)}{|S|} \ge \frac12(d - 2\sqrt{d-1}).
\end{align*}
\end{lemma}

\begin{lemma}\label{lem:ramanujan-graph-property}
For any $a_i \in A$, let $G' = R'(3, Ca_i)$ be the connected graph defined as above, on a set of nodes $V$, say. Let $B \subseteq V$ be such that $|B| \le \ffrac{Ca_i}{2}$. Then, $\delta_{G'}(B) \ge \ffrac{|B|}{100}$. 
\end{lemma}
\begin{proof}
In the original Ramanujan multigraph $R(3, Ca_i)$ (which creates $R'(3, Ca_i)$ after removing repeated edges), the cut $(B, V \setminus B)$ has at least $\ffrac{3|B|}{100}$ edges, using Cheeger's inequality (Lemma \ref{lem:cheegers-inequality-apdx}). In $R'(3, Ca_i)$, the cut size drops by a factor of at most $3$ since we only remove repeated edges, and there can be at most $3$ edges between any two nodes in the graph.
\end{proof}

We construct an instance of {\GHA} by placing these $3m$ $R'$-graphs at unique leaves of a binary tree $B_r$ (see Figure \ref{fig:bounded-degree-reduction-graph}). Like Theorem \ref{thm:bounded-degree-planar-approx-lower-bound}, the binary tree $B_r$ is constructed such that it has at least $3m$ leaves and at most $12m$ nodes. It is easy to see that the smallest balanced binary tree with at least $3m$ leaves satisfies this constraint. Note that the constructed graph has maximum degree $4$.

There are $CmT + |B_r|$ nodes in this graph.
The valuations of the houses are defined the same way as before: for each $j \in [m]$, there are $CT$ houses valued at $j$ and there are $|B_r|$ houses valued at $0$. $C$ again will be decided later. Along with ensuring $C$ is polynomial, we will also ensure that $C$ is even and greater than or equal to $6$.

We show that the minimum total envy output by $\mathsf{ALG}_{\cal G}$ for the house allocation instance is greater than or equal to $\left \lceil (14400\gamma m^3 T + 1)^{\ffrac{1}{\varepsilon}} \right \rceil$ if and only if there exists {\em no} valid partition for the original \textsc{3-Partition} instance i.e. the original instance was a NO instance.

$(\Leftarrow)$ Assume there is a valid $3$-partition for the original instance. This case follows the exact same way as Theorem \ref{thm:trees-approx-lower-bound}. The minimum envy is upper bounded by $3m^2$ and the envy output by $\mathsf{ALG}_{\cal G}$ is upper bounded by $3m^2\gamma(CmT+|B_r|)^{1-\varepsilon} \le 3m^2\gamma(CmT+12m)^{1-\varepsilon}$.

$(\Rightarrow)$ Assume there is no valid $3$-partition in the original instance. We will show that any allocation must have an envy of at least $\ffrac{C}{200}$. 

Similar to Theorem \ref{thm:general-approx-lower-bound}, we refer to a value $j \in [m] \cup \{0\}$ as a majority value of a graph $R'(3, Ca_i)$ if this value has been allocated to at least $\ffrac{Ca_i}{2}$ nodes in the grid. 

For some allocation $\pi$, assume there exists a graph $R'(3, Ca_1)$ with no majority value. Let $Q_j$ be the set of nodes with value $j$ in this graph under allocation $\pi$. Note that the number of edges from $Q_j$ to nodes with a different value is at least $\ffrac{|Q_j|}{100}$ (Lemma \ref{lem:ramanujan-graph-property}). Therefore the total number of edges between nodes with different value is lower bounded by $\ffrac{C}{200}$.
We divide by $2$ since each edge gets counted at most twice. Since each of these edges has an envy of $1$, we are done.

From here on, assume that each $R'$-graph has a majority value. 
Let graph $i$ correspond to $R'(3,Ca_i)$. For notational convenience, assume WLOG that graphs $\{1, \dots, \ell_1\}$ have majority value $1$, $\{{\ell_1 +1}, \dots, {\ell_1 + \ell_2}\}$ have majority value $2$ and so on. Using analysis similar to \Cref{thm:trees-approx-lower-bound}, there exists one $j \in [m]$ such that $\sum_{h \in [\ell_j]}a_{\ell_1 + \dots + \ell_{j-1} +h} > T$. Assume again for notational convenience that $j = 1$. Since all these values are integers, we can restate it as $\sum_{h \in [\ell_1]}a_{h} \ge T + 1$. This implies $\sum_{h \in [\ell_1]} Ca_{h} \ge CT + C$. 

Coming to our allocation $\pi$, $\sum_{h \in [\ell_1]} Ca_{h} \ge CT + C$ implies that there are at least $C$ nodes in graphs $\{1,2, \dots, \ell_1\}$ which are not allocated a value of $1$. Let $A_i$ correspond to the set of nodes in grid $i$ allocated a value other than $1$. Note that $\sum_{i \in [\ell_1]} |A_i| \ge C$. For each $i$, the number of edges between nodes in $A_i$ and nodes with value $1$, is at least $\ffrac{|A_i|}{100}$ (Lemma \ref{lem:ramanujan-graph-property}). This gives us a lower bound of $\ffrac{C}{100}$ edges with envy at least $1$.

We set $C = 200\left \lceil (14400\gamma m^3 T + 1)^{\ffrac{1}{\varepsilon}} \right \rceil$ to complete the reduction; this setting crucially ensures that each $Ca_i$ is even and at least $6$. When there is no valid $3$-partition, the total minimum envy (and therefore, the envy output by $\mathsf{ALG}_{\cal G}$) is at least $\ffrac{C}{200} = \left \lceil (14400\gamma m^3 T + 1)^{\ffrac{1}{\varepsilon}} \right \rceil$. However, when there is a valid $3$-partition, the envy output by $\mathsf{ALG}_{\cal G}$ is strictly upper bounded by:
\begin{align*}
    3m^2\gamma(CmT+12m)^{1-\varepsilon} \le 36m^3 \gamma TC^{1-\varepsilon} < 7200m^3 \gamma T \left ( \left \lceil (14400\gamma m^3 T + 1)^{\ffrac{1}{\varepsilon}} \right \rceil \right )^{1 - \varepsilon} < \left \lceil (14400\gamma m^3 T + 1)^{\ffrac{1}{\varepsilon}} \right \rceil.
\end{align*}
This concludes the proof.
\end{proof}

\end{document}